\documentclass[a4paper,11pt,centertags,psamsfonts]{amsart}
\usepackage{times,epsfig,amssymb}

\numberwithin{equation}{section}

\DeclareMathAlphabet{\mathsf}{OT1}{phv}{m}{n}
\DeclareMathAlphabet{\mathrm}{OT1}{ptm}{m}{n}

\DeclareSymbolFont{ER}{U}{eur}{m}{n} 
\DeclareSymbolFont{SY}{U}{psy}{m}{n}
      
\DeclareMathSymbol{,}{\mathpunct}{SY}{'054}
\DeclareMathSymbol{.}{\mathpunct}{SY}{'056}       
\DeclareMathSymbol{:}{\mathpunct}{SY}{'072}

\DeclareMathSymbol{(}{\mathopen}{SY}{'050}
\DeclareMathSymbol{)}{\mathclose}{SY}{'051}

\DeclareMathSymbol{+}{\mathbin}{SY}{'053}
\DeclareMathSymbol{-}{\mathbin}{SY}{'055}
\DeclareMathSymbol{=}{\mathbin}{SY}{'075}

\DeclareMathSymbol{<}{\mathbin}{SY}{'074}
\DeclareMathSymbol{>}{\mathbin}{SY}{'076}
\DeclareMathSymbol{\leq}{\mathbin}{SY}{'243}
\DeclareMathSymbol{\geq}{\mathbin}{SY}{'263}
\DeclareMathSymbol{\nneq}{\mathbin}{SY}{'271}
\DeclareMathSymbol{\in}{\mathbin}{SY}{'316}
\DeclareMathSymbol{\nnotin}{\mathbin}{SY}{'317}
\DeclareMathSymbol{\times}{\mathbin}{SY}{'264}
\DeclareMathSymbol{\pm}{\mathbin}{SY}{'261}
\DeclareMathSymbol{\subset}{\mathbin}{SY}{'314}
\DeclareMathSymbol{\supset}{\mathbin}{SY}{'311}
\DeclareMathSymbol{\subseteq}{\mathbin}{SY}{'315}
\DeclareMathSymbol{\supseteq}{\mathbin}{SY}{'312}

\DeclareMathSymbol{/}{\mathord}{SY}{'057}
\DeclareMathSymbol{\ast}{\mathord}{SY}{'052}
\DeclareMathSymbol{\perp}{\mathord}{SY}{'136}

\renewcommand{\neq}{\nneq}

\renewcommand{\notin}{\nnotin}

\renewcommand{\theequation}{\arabic{section}.\arabic{equation}}

\newcommand{\R}{\mathbb{R}}
\newcommand{\C}{\mathbb{C}}

\newcommand{\cC}{\mathcal{C}}
\newcommand{\cD}{\mathcal{D}}
\newcommand{\cE}{\mathcal{E}}

\newcommand{\cG}{\mathcal{G}}
\newcommand{\cH}{\mathcal{H}}
\newcommand{\cI}{\mathcal{I}}

\newcommand{\cM}{\mathcal{M}}

\newcommand{\cO}{\mathcal{O}}

\newcommand{\cS}{\mathcal{S}}
\newcommand{\cT}{\mathcal{T}}

\newcommand{\cV}{\mathcal{V}}
\newcommand{\cW}{\mathcal{W}}
\renewcommand{\Im}{{\ensuremath{\mathrm{Im}}}}
\renewcommand{\Re}{{\ensuremath{\mathrm{Re}}}}
\newcommand{\supp}{{\ensuremath{\mathrm{supp}}}}

\newcommand{\diag}{{\ensuremath{\mathrm{diag}}}}
\newcommand{\pdist}{{\ensuremath{\mathrm{pdist}}}}
\newcommand{\tr}{\mathrm{tr}}

\newcommand{\sign}{\mathrm{sign}}

\newcommand{\fD}{\mathfrak{D}}
\newcommand{\fd}{\mathfrak{d}}
\newcommand{\fF}{\mathfrak{F}}

\newcommand{\fI}{\mathfrak{I}}

\DeclareMathOperator{\Ran}{\mathrm{Ran}}
\DeclareMathOperator{\Ker}{\mathrm{Ker}}

\newcommand{\1}{\mathbb{I}}
\newcommand{\spec}{{\ensuremath{\rm spec}}}

\renewcommand{\det}{\mathrm{det\ }}

\newcommand{\sk}{\mathsf{k}}
\newcommand{\ii}{\mathrm{i}}
\newcommand{\e}{\mathrm{e}}

\newcommand{\bw}{\mathbf{w}}

\DeclareMathSymbol{\emptyset}{\mathord}{SY}{'306}
\DeclareMathSymbol{\oplus}{\mathord}{SY}{'305}

\newtheorem{theorem}{Theorem}{\bf}{\it}
\newtheorem{proposition}[theorem]{Proposition}{\bf}{\it}
\newtheorem{corollary}[theorem]{Corollary}{\bf}{\it}
{\bf}{\it}
\newtheorem{example}[theorem]{Example}{\it}{\rm}
\newtheorem{lemma}[theorem]{Lemma}{\bf}{\it}
\newtheorem{remark}[theorem]{Remark}{\it}{\rm}
{\bf}{\it}
{\bf}{\it}
{\bf}{\it}

\title[Finite propagation speed and free quantum fields on networks]
{Finite propagation speed and causal free quantum fields on networks}

\author[R. Schrader]{Robert Schrader}

\address{Robert Schrader\\ Institut f\"{u}r
Theoretische Physik\\ Freie Universit\"{a}t Berlin, Arnimallee 14
\\ D-14195 Berlin,
Germany}
\email{robert.schrader@fu-berlin.de}

\keywords{Metric graphs, Klein-Gordon equation, Wave Equation, Quantum field theory}

\subjclass{34B45,35L05,35L20}

\begin{document}

\begin{abstract}
Laplace operators on metric graphs give rise to Klein-Gordon and wave operators. 
Solutions of the Klein-Gordon equation and the wave equation are studied and 
finite propagation speed is established. Massive, free quantum fields are then constructed, whose 
commutator function is just the Klein-Gordon kernel. As a consequence of finite propagation speed 
Einstein causality (local commutativity) holds. Comparison is made with an alternative construction 
of free fields involving RT-algebras. 

\vspace{0.3cm}
PACS: 03.65.Nk, 03.70.+k, 73.21.Hb
\end{abstract}

\maketitle

\section{Introduction}
In the last years the study of quantum systems on networks has received an increasing attention. 
They are of interest for possible applications in condensed matter physics. 
In addition interesting mathematical structures appear giving rise to a host of attractive problems, 
see e.g. the articles in \cite{EKKST} and further references given there. 
In this article we study Klein-Gordon and wave equations on any metric graph and for any given
Laplace operator thereon. We establish existence, uniqueness and finite propagation speed for given 
initial data.
In addition we construct free quantum fields on arbitrary metric graphs. The construction of such 
fields was initiated in \cite{BellMin,BellBurrMinSo,BeMiSoI}. The results obtained there were 
applied to a study of spin transport and conductance \cite{BellBurrMinSo,BeMiSoI,Ragoucy}, incorporating
additional techniques developed in \cite{BellBurrMinSo,BeMiSoII}. The main tool for the construction 
of these fields was the use of (a simple version of) \emph{RT-algebras} \cite{CMRS,MRSI,MRSII}.
Also the construction there was limited to relatively simple graphs. 
The construction we present here does not involve RT-algebras and uses only standard and familiar 
methods of second quantization. However, we will be able to relate our construction 
to the RT-construction. Spin will not be considered. In order to avoid dealing with infrared problems, 
we will construct only massive and not massless quantum fields. 

We briefly outline our strategy and our results. As a starting point we choose the Hilbert 
space of square integrable functions on the graph as the 1-particle space. Next we make a choice of 
a self-adjoint Laplacian $-\Delta$ on the graph, which is not necessarily positive. However, 
$-\Delta$ will always be bounded below. To define $-\Delta$, we follow the discussion in \cite{KS1} 
by specifying boundary conditions at the vertices of the graph for the operator 
given as the second derivative acting on functions on the graph. 
Given the Laplacian and a mass $m>0$ and motivated by relativistic quantum theory, we introduce the 
energy operator $\sqrt{-\Delta+m^2}$, the d'Alembert operator (wave operator) $\Box=\partial_t^2-\Delta$ 
and the Klein-Gordon operator $\Box+m^2$. \footnote{We work in units where $\hbar=c=1$}
Unique solutions of the classical Klein-Gordon equation for given Cauchy data are then obtained by 
using 
\begin{equation}\label{wavekernel}
\frac{\sin\sqrt{-\Delta+m^2}t}{\sqrt{-\Delta+m^2}},
\end{equation}
which is the \emph{Klein-Gordon kernel} for $m>0$ and the \emph{wave kernel} for $m=0$ and which 
will be studied in detail. In particular finite propagation speed will be established. This notion 
makes sense, since on a metric graph the distance between two points is well defined, so the concepts of 
two \emph{events}, that is points in  space-time, being space-like separated makes sense. 
Finite propagation speed for solutions of the wave 
equation on smooth manifolds is well studied and understood, see e.g 
\cite{CheegerGromovTaylor,Evans,Taylor,TaylorI}.
So far for spaces with singularities finite propagation speed has been proved only for the case when 
the singularities are conical \cite{CheegerTaylor}.

Applying second quantization and a given choice of the Laplacian, we arrive at free fields which 
satisfy the Klein-Gordon equation.  They are hermitian as soon as the boundary conditions defining the 
Laplacian are chosen to be real, a notion that will be explained below and which is equivalent to 
time reversal invariance in quantum mechanics, when the Laplace operator is taken to be a Schr\"odinger 
operator. As usual, the non-hermitian scalar fields carry charge. 
For their construction we work with two Laplacians, one for the particle and the other one for the 
antiparticle. They are such that the boundary conditions defining them are the complex conjugates of 
each other, again a notion that will be explained in due time.  
Since the commutator is actually given by the kernel \eqref{wavekernel}, Einstein causality 
(local commutativity) is just another formulation of finite propagation speed. In other words, 
we show that the commutator vanishes for space-like separated events. 
Our proof is different from the standard proof of finite propagation speed on smooth manifolds.
Our methods, however, do not allow us to prove finite propagation speed and hence 
Einstein causality in full generality. As a matter of fact, we miss those space-like separated 
events, whose space components are both points in the interior of the 
graph. Theorem \ref{theo:mfinite} gives the precise conditions and statements. 

The article is organized as follows. In Section \ref{sec:laplace} we summarize several properties 
of Laplace operators on metric graphs in a form needed for the next sections. 
It includes a detailed discussion of their (improper) eigenfunctions. In fact since these eigenfunctions 
give us the integral kernel of the Klein-Gordon kernel, some of their properties are crucial for 
establishing finite propagation speed. In addition to recalling several results from 
\cite{KPS1,KS1,KS2,KS8}, we also establish new and relevant ones. This includes the following. 
Viewing the Laplacian as the Hamiltonian of a quantum dynamical system,
there is an associated scattering theory. As it turns out, the on-shell scattering matrix enters the 
eigenfunctions \cite{KS1} and hence also the integral kernel of the Klein-Gordon kernel. 
The crucial ingredients in proving finite propagation speed are the analytic properties of the S matrix.
In the single vertex case the information we gain on the S matrix is so detailed, that we are able to 
establish finite propagation speed even in the case that the Laplacian has bound states.

In Section \ref{sec:clKG} we discuss classical solutions of the Klein-Gordon and the wave equation. 
There we also formulate the finite propagation speed result, the proof of which is given in 
Appendix \ref{app:KGsupport}. In Section \ref{sec:relnew} we construct 
space-time dependent relativistic free fields, both hermitian and non-hermitian, 
that satisfy the Klein-Gordon equation and the same boundary conditions as those for the given Laplacian. 
There we also show, that their commutator function equals (minus) the Klein-Gordon kernel 
\eqref{wavekernel}. The proof of the orthonormality of the improper eigenfunctions of the Laplacian 
is given in Appendix \ref{app:orthoproof}. 

\section{Laplace Operators on Metric Graphs , their spectral properties and their 
eigenfunctions}\label{sec:laplace}
In this section and for the convenience of the reader, we recall the construction of self-adjoint 
Laplace operators on metric graphs in terms of boundary conditions. Also we list several of 
their properties, in particular of 
their eigenfunctions. They will be needed when we establish finite propagation speed and 
when we construct free fields and discuss some of their properties.
We start with some elementary concepts from graph theory. The material is mainly taken from \cite{KPS1}.
\vspace{0.3cm}
\subsection{Basic concepts}~~~\\
A finite graph is a 4-tuple $\cG=(\cV,\cI,\cE,\partial)$, where $\cV$ is a
finite set of \emph{vertices}, $\cI$ a finite set of \emph{internal
edges} and $\cE$ a finite set of \emph{external edges}. Elements in
$\cI\cup\cE$ are called \emph{edges}. $\partial$ is a map, which assigns to each
internal edge $i\in\cI$ an ordered pair of (possibly equal) vertices
$\partial(i):=\{v_1,v_2\}$ and to each external edge $e\in\cE$ a single
vertex $v$. The vertices $v_1=:\partial^-(i)$ and $v_2=:\partial^+(i)$ are
called the \emph{initial} and \emph{final} vertex of the internal edge
$i$, respectively. The vertex $v=\partial(e)$ is the initial vertex of the
external edge $e$. If $\partial(i)=\{v,v\}$, that is
$\partial^-(i)=\partial^+(i)$, then $i$ is called a \emph{tadpole}. 
A graph is \emph{compact} if $\cE=\emptyset$, otherwise it is
\emph{noncompact}. 
Two vertices $v$ and $v^\prime$ are called \emph{adjacent} if there is an
internal edge $i\in\cI$ such that $v\in\partial(i)$ and
$v^\prime\in\partial(i)$. A vertex $v$ and the (internal or external) edge
$j\in\cI\cup\cE$ are \emph{incident} if $v\in\partial(j)$.

We do not require the map $\partial$ to be injective. In particular, any two
vertices are allowed to be adjacent to more than one internal edge and two
different external edges may be incident with the same vertex. If
$\partial$ is injective and $\partial^-(i)\neq\partial^+(i)$ for all
$i\in\cI$, the graph $\cG$ is called \emph{simple}.
The \emph{degree} $\deg(v)$ of the vertex $v$ is defined as
\begin{equation*}
\deg(v)=|\{e\in\cE\mid\partial(e)=v\}|+|\{i\in\cI\mid\partial^-(i)
=v\}|+|\{i\in\cI\mid\partial^+(i)=v\}|,
\end{equation*}
that is, it is the number of (internal or external) edges incident with the
given vertex $v$ and by which every tadpole is counted twice.
A vertex is called a \emph{boundary vertex} if it is incident with at least one
external edge. The set of all boundary vertices will be denoted by $\partial
\cV$ such that $|\partial\cV|\le |\cE|$ holds.
The vertices not in $\partial \cV$, that is in 
$\cV_{\mathrm{int}}=\cV\setminus\partial \cV$ are
called \emph{internal vertices}.

The compact graph $\cG_{\mathrm{int}}=(\cV,\cI,\emptyset,\partial|_{\cI})$
will be called the \emph{interior} of the graph $\cG=(\cV,\cI,$ $
\cE,\partial)$. It is obtained from $\cG$ by eliminating all external
edges $e$. 
Correspondingly, if $\cE\neq \emptyset$, the noncompact graph 
$\cG_{\mathrm{ext}}=(\partial\cV,\emptyset,\cE,\partial|_{\cE})$
is called the \emph{exterior} of $\cG$. We will view both $\cG_{\mathrm{int}}$ and $\cG_{\mathrm{ext}}$ as 
subgraphs of $\cG$ with $\cG_{\mathrm{int}}\cap\cG_{\mathrm{ext}}=\partial\cV$.

Throughout the whole work we will from now on assume that the graph $\cG$ is connected,
that is, for any $v,v^\prime\in \cV$ there is an ordered sequence $\{v_1=v,
v_2,\ldots, v_{n-1}, v_n=v^\prime\}$ such that any two successive vertices in this
sequence are adjacent. In particular, this implies that any vertex of the
graph $\cG$ has nonzero degree, that is for any vertex there is at least one
edge with which it is incident. $\cG_{\mathrm{int}}$ is connected if $\cG$ is.
For connected $\cG$, the graph $\cG_{\mathrm{ext}}$ is connected if and only if $\partial\cV$ 
consists of one vertex only. 
By definition a \emph{single vertex graph} is a connected graph which has no internal edges, 
only one vertex, and at least one external edge. 
The \emph{star graph} $\cS(v)\subseteq \cE\cup\cI$ associated to the vertex $v\in \cV$
consists of the set of the edges adjacent to $v$ and of the vertex $v$.

We will endow the graph with the following metric structure. Any internal
edge $i\in\cI$ will be associated with an interval $I_i=[0,a_i]$ with $a_i>0$
such that the initial vertex of $i$ corresponds to $x=0$ and the final
one to $x=a_i$. The open interval $I_i^o=(0,a_i)$ will be called the interior of the edge $i$.
We call the number $a_i$ the length of the internal edge $i$.
Any external edge $e\in\cE$ will be associated with a
semi-line $I_e=[0,+\infty)$ whose interior is $I_e^o=(0,+\infty)$. 
The set of lengths $\{a_i\}_{i\in\cI}$, which will also be treated
as an element of $\R^{|\cI|}$, will be denoted by $\underline{a}$. 
A compact or noncompact graph $\cG$ endowed with a metric structure is called
a \emph{metric graph} $(\cG,\underline{a})$. For the purpose of a compact notation 
we set $a_e=\infty$ for $e\in\cE$. The metric structure induces a distance function 
$d(p,q)\ge 0$ with the familiar three properties 
\begin{itemize}
\item{$d(p,p)=0$}
\item{$d(p,q)=d(q,p)$}
\item{$d(p,q)\le d(p,p^{\prime})+d(p^{\prime},q)$}
\end{itemize}
for all $p,p^{\prime},q\in\cG$. This defines a topology on $(\cG,\underline{a})$, such that 
$d(p,q)$ is continuous in both variables.
For any $e,e^\prime\in\cE$ we call $\pdist(e,e^\prime)=d(\partial(e),\partial(e^\prime))$ 
the \emph{passage distance} from the external edge $I_e$ to the external edge $I_{e^\prime}$.
Thus $\pdist(e,e^\prime)=0$ if and only if $\partial(e)=\partial(e^\prime)$ and 
$\pdist(e,e^\prime)\ge \min_{i\in\cI}a_i>0$, whenever $\partial(e)\neq\partial(e^\prime)$.
$d(p,q)\ge \pdist(e,e^\prime)$ holds for any $p\in I_e$ and $q\in I_{e^\prime}$. 

On the graph $\cG$ there is a natural Lebesgue measure $dp$. 
In particular there is the 
Hilbert space $L^2(\cG)$ of square integrable functions on $\cG$. We write the scalar product as
\begin{equation}\label{scalarpconv}
\langle \psi,\phi\rangle_\cG=\int_{\cG}\overline{\psi(p)}\phi(p)dp 
\end{equation}
or simply $\langle \psi,\phi\rangle$, if the context is clear.
We write $x\in I_j=[0,a_j]$ for the coordinate of the point $p\in \cG$ 
if $p$ lies on the edge $j\in\cE\cup\cI$ at the point $x$  and we shall
say that the pair $(j,x)$ is the 
\emph{local coordinate} for $p$. For short and whenever convenient we will also view $(j,x)$ as a point in 
$\cG$. 
A complex valued function on the graph, or more precisely on $\cG\setminus\cV$,
may be considered to be a family
$\psi=\{\psi_{j}\}_{j\in\cE\cup\cI}$ of complex valued functions $\psi_{j}$
defined on $(0,a_j)$, so by the convention just made $\psi(j,x)=\psi_j(x)$.
With this notation the scalar product may be written as 
\begin{equation*}
\langle \psi,\phi\rangle=\sum_{j\in\cE\cup\cI}\int_0^{a_j}\overline{\psi_j(x)}
\phi_j(x)dx
\end{equation*}
Also we define the derivative $\psi^\prime=\partial_x\psi$ of $\psi$ as
$$
(\psi^\prime)_j(x)=\frac{d}{dx}\psi_j(x).
$$
We also introduce the following set of boundary values of $\psi$ and its derivative
as
\begin{equation}\label{bdvpsi}
\underline{\psi} = \begin{pmatrix} \{\psi_e(0)\}_{e\in\cE} \\
                                   \{\psi_i(0)\}_{i\in\cI} \\
                                   \{\psi_i(a_i)\}_{i\in\cI} \\
                                     \end{pmatrix},\qquad
\underline{\psi}' = \begin{pmatrix} \{\psi_e'(0)\}_{e\in\cE} \\
                                   \{\psi_i'(0)\}_{i\in\cI} \\
                                   \{-\psi_i'(a_i)\}_{i\in\cI} \\
                                     \end{pmatrix}.
\end{equation}
The ordering of the set $\cE$ is arbitrary but fixed as is the ordering in $\cI$. 
Given an ordering, in \eqref{bdvpsi} the boundary values on the external edges come 
first, then the boundary values at the initial vertices and finally the boundary 
values at the final vertices. 
Note also that $\underline{\psi}'$ is defined in terms of the inward normal derivative, 
which is intrinsic, that is independent of the special choice of the orientation 
on each of the internal edges.

The Laplace operator is defined as
\begin{equation*}
\left(-\Delta_{A,B}\psi\right)_j (x) = -\frac{d^2}{dx^2}
\psi_j(x),\qquad j\in\cI\cup\cE
\end{equation*}
with boundary conditions
\begin{equation}\label{bdycond}
A\underline{\psi} + B\underline{\psi}' = 0.
\end{equation}
$A$ and $B$ are $(|\cE|+2|\cI|)\times(|\cE|+2|\cI|)$ matrices. For later reference we rewrite this 
condition as 
\begin{equation}\label{bdycondrewr}
(A,B)\begin{pmatrix}\underline{\psi}\\\underline{\psi}'\end{pmatrix}=0,
\end{equation}
where $(A,B)$ is the $(|\cE|+2|\cI|)\times2(|\cE|+2|\cI|)$ matrix obtained by putting the matrices 
$A$ and $B$ next to each other. So \eqref{bdycondrewr} is the condition 
\begin{equation}\label{bdycondrewr1}
\begin{pmatrix}\underline{\psi}\\\underline{\psi}'\end{pmatrix}\in \Ker (A,B).
\end{equation}
The operator
$-\Delta_{A,B}$ is self-adjoint if and only if the matrix 
$(A,B)$ has maximal rank and the matrix $AB^\dagger$ is hermitian. Obviously 
for any invertible $C$ the pair $(CA,CB)$ gives the same boundary conditions since 
$\Ker (CA,CB)=\Ker (A,B)$. Moreover, with these conditions $\Ker (A,B)$ is a 
maximal isotropic subspace $\cM(A,B)$ w.r.t. the canonical hermitian symplectic form on 
$\cC^{2(|\cE|+2|\cI|)}$ and all hermitian subspaces can be written in this form, see \cite{KS1}. 
Moreover $\cM(A,B)=\cM(A^\prime,B^\prime)$ if 
and only if $A^\prime=CA,B^\prime=CB$  for some invertible $C$.
For a detailed discussion concerning the self-adjointness see \cite{KS1,KS8}.
In addition, if the pair $(A,B)$ satisfies these two conditions, so does the complex conjugate 
pair $(\bar{A},\bar{B})$ giving rise to the Laplacian $\Delta_{\bar{A},\bar{B},
\underline{a}}$. Let $n_+(AB^\dagger)$ be the number of positive eigenvalues 
of $AB^\dagger$, counting multiplicities. The identity 
\begin{equation}\label{nplusbar}
n_+(AB^\dagger)=n_+(\bar{A}\bar{B}^\dagger)
\end{equation}
is clear. In fact, $AB^\dagger$ and $\bar{A}\bar{B}^\dagger$ actually have the same spectrum.
\begin{proposition}\label{prop:spectrum}
The absolute continuous spectrum of each $-\Delta_{A,B}$ is the 
interval $[0,\infty)$. It has multiplicity equal to the number of external edges, $|\cE|$.
The number of negative eigenvalues, counting multiplicities, is at most 
$n_+(AB^\dagger)\;(\;\le |\cE|+2|\cI|\;)$. 
It is equal to $n_+(AB^\dagger)$ if $\cI=\emptyset$.
\end{proposition}
Below we shall see that the external edges provide a natural labeling for the multiplicities of the
absolutely continuous spectrum.
\begin{proof} 
We claim that all Laplacians $-\Delta_{A,B}$ are finite 
rank perturbations of each other, that is the difference of two resolvents is always a finite rank 
operator. To see this, consider the Hilbert space
\begin{equation}
\cH =\cH(\cE,\cI, \underline{a})= \cH_{\cE}\oplus\cH_{\cI},\qquad \cH_{\cE} = 
\oplus_{e\in\cE}\cH_e,\; \cH_{\cI} = \oplus_{i\in\cI}\cH_i,
\end{equation}
where $\cH_e = L^2([0,\infty),dx)$ for all $e\in\cE$ and $\cH_i = L^2([0, a_i],dx )$ for all $i\in\cI$.
Then $L^2(\cG)\cong\cH$. By $\cD_j$ with $j\in \cE\cup \cI$ 
denote the set of all $\psi_j\in \cH_j$ such that $\psi_j (x)$ and its derivative $\psi^\prime_j (x)$
are absolutely continuous and $\psi_j (x)$ is square integrable. 
Let $\cD_j^0$ denote the subset of consisting of elements $\psi_j$ which satisfy
\begin{align*}
\psi_j(0)&=\psi^\prime(0)=0\quad\quad\qquad\qquad\qquad\mbox{when}\quad j\in\cE\\
\psi_j(0)&=\psi^\prime(0)=\psi_j(a_j)=\psi^\prime(a_j)=0\quad\mbox{when}\quad j\in\cI.
\end{align*}
Let $\Delta^0$ be defined as the second derivative operator, $\Delta^0\psi=\psi^{\prime\prime}$, with domain
$$
\cD^0=\oplus_{j\in\cE\cup\cI}\cD^0_j\subset \cH.
$$
Then the deficiency index of $-\Delta^0$ is equal to $(|\cE|+2|\cI|,|\cE|+2|\cI|)$
and every self-adjoint extension is of the form $-\Delta_{A,B}$ for a suitable boundary 
condition $(A,B)$. Thus the claim follows by general results on self-adjoint extensions, see, e.g., 
Appendix A in \cite{AGHKH} and the references quoted there. 
The last statement is just Theorem 3.7 in \cite{KS9}. 
\end{proof}
We elaborate on the sufficient criterion $n_+(AB^\dagger)=0$ for the absence of negative eigenvalues.
For given boundary condition $(A,B)$ introduce the meromorphic matrix valued function 
in $\sk$
\begin{equation}\label{mathfracs}
\mathfrak{S}(\sk;A,B)=-(A+\ii \sk B)^{-1}(A-\ii \sk B).
\end{equation}
Observe that $\mathfrak{S}(\sk;CA,CB)=\mathfrak{S}(\sk;A,B)$ holds for all invertible 
$C$, so this function depends  only on the maximal isotropic subspace defined by $(A,B)$, 
$\mathfrak{S}(\sk;A,B)=\mathfrak{S}(\sk;\cM(A,B))$.
\begin{lemma}(\cite{KS1}, Theorem 2.1;\cite{KS8},Theorem 3.12,\cite{KS9}; 
Theorem 3.7)\label{lem:mathfracs}
$\mathfrak{S}(\sk;A,B)$ exists and is unitary for all $\sk>0$. 
Its poles lie on the imaginary axis. There are no poles on the positive imaginary axis
if and only if $AB^\dagger \le 0$ and then $-\Delta_{A,B}$ has no 
negative eigenvalues.
\end{lemma}
The condition $A^\dagger B\le 0$ has the following local formulation, see Definition 
2.3 in \cite{KS9}, in terms of vertex quantities and which will be used below.
By Proposition 4.2 in \cite{KS8} for given boundary conditions $(A,B)$ 
there is an invertible $C$ such that the two matrices $CA$ and $CB$ have a common 
block decomposition
\begin{equation}\label{block}
CA=\bigoplus_{v\in\cV}A(v)\qquad CB=\bigoplus_{v\in\cV}B(v)
\end{equation}
where the pair $(A(v),B(v))$ gives the boundary conditions at the vertex $v$.
Thus we obtain the 
\begin{lemma}\label{lem:decomp}
The following block decomposition holds for all $\sk$
\begin{equation}\label{block1}
\mathfrak{S}(\sk;A,B)=\bigoplus_{v\in\cV}\mathfrak{S}(\sk;A(v),B(v)).
\end{equation}
In particular, if the boundary conditions $(A,B)$ are such that $AB^\dagger\le 0$, 
then 
$A(v) B(v)^\dagger\le 0$ holds for all vertices $v$ and therefore no
$\mathfrak{S}(v;\sk)=\mathfrak{S}(\sk;A(v),B(v))$ has poles on the positive 
imaginary axis.
\end{lemma}
With the notation just introduced there is the following characterization of $\sk$-
independence.
\begin{lemma}\cite{KPS1}
$\mathfrak{S}(\sk;A,B)$ is $\sk$-independent if and only if $AB^\dagger=0$ and hence 
if and only if $A(v) B(v)^\dagger=0$ holds for all $v\in\cV$.
\end{lemma}
Alternative characterizations of such boundary conditions are given in \cite{KS8}, Remark 3.9
and \cite{KPS1}, proposition 2.4. Thus in the single vertex case all $\sk$-independent 
S-matrices are of the form 
\begin{equation}\label{kindeps}
S=\1-2P
\end{equation}
with $P$ being an orthogonal projector and then $S^{-1}=S^\dagger=S$ holds.
In combination with theorem 3.7 in \cite{KS9} 
$-\Delta_{A,B}\ge 0$ follows for such boundary conditions, see also Lemma \ref{lem:mathfracs}.

The boundary conditions actually fix the graph. More precisely,
given  finite intervals $I_i\;(i\in\cI)$ and half lines $I_e\;(e\in\cE)$, 
and functions $\psi=\{\psi_j\}_{j\in \cI\cup\cE}$ on them, generically the boundary condition 
\eqref{bdycond} given by the pair $(A,B)$ uniquely fixes the graph $\cG$ with a maximal 
set of vertices, such that the boundary conditions are {\it local}, see \cite{KS1,KS8}
for details.
 
For given $l\in\cE$ consider the following solution $\psi^{l}(\sk)$ of
the stationary Schr\"{o}dinger equation at energy $\sk^2>0$,
\begin{equation}\label{schroedinger}
-\Delta_{A,B}\psi^{l}(\,;\sk)=\sk^2\psi^{l}(\,;\sk)
\end{equation}
and of the form
\begin{equation}\label{psidef}
\psi^{l}_{j}(x;\sk)=\begin{cases} 
  \e^{-\ii\sk x}\delta_{jl} + S(\sk)_{jl} \e^{\ii\sk x} & \text{for}\;j\in\cE \\
&\\
                                  \alpha(\sk)_{jl} \e^{\ii\sk x}+
        \beta(\sk)_{jl} \e^{-\ii\sk x} & \text{for}\; j\in\cI. \end{cases}
\end{equation}
So intuitively we are looking at what happens to an incoming plane wave 
$\e^{-\ii\sk x}$ in channel $l$ when it moves through the graph. Observe that choosing the Laplacian
$-\Delta_{A,B}$ as Schr\"{o}dinger operator, quantum mechanically this means that we have free motion away 
from the vertices.  The vertices in turn act as \emph{beam splitters} in a way described by the boundary 
condition $(A,B)$.

The number $S(\sk)_{jl}$ for $j\neq l$ is the \emph{transmission amplitude}
from channel $l\in\cE$ to channel $j\in\cE$ and $S(\sk)_{ll}$ is the
\emph{reflection amplitude} in channel $l\in\cE$. So their absolute value squares may be
interpreted as transmission and reflection probabilities, respectively. 
The elements $S(\sk)_{jl}$ combine to form the scattering matrix 
$$
S(\sk)=S_{A,B}(\sk).
$$
The ``interior'' amplitudes $\alpha(\sk)_{jl}=\alpha_{A,B}(\sk)_{jl}$ 
and $\beta(\sk)_{jl}=\beta_{A,B}(\sk)_{jl}$ are
also of interest, since they describe how an incoming wave moves through a
graph before it is scattered into an outgoing channel. 

The condition that $\psi^{l}(\,;\sk)$ satisfies the boundary condition leads to the 
solution
\begin{equation}\label{defs}
\begin{pmatrix} S(\sk)\\
                      \alpha(\sk)\\
                    \beta(\sk)\end{pmatrix}
                      =-Z(\sk)^{-1}(A-\ii\sk B)
              \begin{pmatrix} \1\\
                               0\\
                               0 \end{pmatrix}
\end{equation}
with the matrices
\begin{align}\label{zdef}
Z(\sk)&=Z_{A,B}(\sk)\,= A X(\sk)+\ii\sk B Y(\sk)\\\nonumber
X(\sk)&=X(\sk;\underline{a})\quad=\begin{pmatrix}\1&0&0\\
                                  0&\1&\1\\
               0&\e^{\ii\sk\underline{a}}&\e^{-\ii\sk\underline{a}}
               \end{pmatrix}\\\nonumber
Y(\sk)&=Y(\sk;\underline{a})\quad=\begin{pmatrix}\1&0&0\\
                                  0&\1&-\1\\
               0&-\e^{\ii\sk\underline{a}}&\e^{-\ii\sk\underline{a}}
               \end{pmatrix}.
\end{align}
The diagonal $|\cI|\times |\cI|$ matrices $\e^{\pm \ii\sk\underline{a}}$
are given by
\begin{equation*}
\e^{\pm \ii\sk\underline{a}}_{\;\quad jk}=\e^{\pm \ii\sk a_{j}}\delta_{jk}\quad
                       \text{for}\quad j,k\in\;\cI.
\end{equation*}
By construction $Z(\sk;A,B,\underline{a})$ is entire in $\sk\in\C$. 
For Neumann boundary conditions the scattering 
is trivial, $S_{A=0,B=\1}(\sk)=\1$.

The $\psi^l(\,;\sk)$ are not in $L^2(\cG)$ , but rather improper eigenfunctions. 
Their main properties are collected in 
\begin{proposition}\label{prop:indep}
For fixed $\sk>0$ the $\psi^l(\,;\sk)$ are linearly 
independent. Any function $\psi$ on $\cG$ satisfying 
$-\Delta_{A,B}\psi=\sk^2\psi$ is a linear combination of these 
$\psi^l(\,;\sk)$, provided $\sk^2$ is not a discrete eigenvalue of 
$-\Delta_{A,B}$.
\end{proposition}
The proof will be given in a moment. The next proposition will be play an important r\^ole in our 
construction of free quantum fields on the graph $\cG$. Set 
\begin{equation}\label{sigma>}
\Sigma^>=\Sigma^>_{A,B}=\{\sk >0\mid \det Z_{A,B}(\sk)=0\}.
\end{equation}
\begin{proposition}\label{prop:ortho} 
The improper eigenfunctions $\psi^l(\,;\sk)$ satisfy the the following orthogonality 
relations
\begin{equation}\label{ortho}
\langle \psi^l(\,;\sk), \psi^{l^\prime}(\,;\sk^\prime)\rangle
=2\pi\delta_{l,l^\prime}\delta(\sk-\sk^\prime)\qquad \sk,\sk^\prime
\in\R_+\setminus \Sigma^>.
\end{equation}
For any $\sk\in\R_+\setminus \Sigma^>$ they span the space associated to the absolutely 
continuous spectrum and so the 
multiplicity of the absolute continuous spectrum equals $|\cE|$.
In particular, if there are no discrete eigenvalues, then 
the $\psi^l(\,;\sk)$ form a
complete set of improper eigenfunctions of $-\Delta_{A,B}$ in $L^2(\cG)$.
\end{proposition}
That there are no discrete eigenvalues means that i) $-\Delta_{A,B}\ge 0$, ii) there are 
no positive eigenvalues and iii) zero is not an eigenvalue. 
The proof of \eqref{ortho} will be given in Appendix \ref{app:orthoproof}.
The remainder follows from the previous proposition.
Recalling the notational convention \eqref{scalarpconv}, \eqref{ortho} reads as
\begin{equation}\label{ortho1}
\int_\cG \overline{\psi^l(p;\sk)}\,\psi^{l^\prime}(p;\sk^\prime)\,dp=2\pi
\delta_{l,l^\prime}\delta(\sk-\sk^\prime).
\end{equation}
For the proof we will need a result concerning the existence of positive (= embedded) eigenvalues. 
\begin{theorem}(\cite{KS1}, Theorem 3.1,\cite{KS9}, Lemma 3.1)\label{theo:posev}
$-\Delta_{A,B}$ 
has a positive eigenvalue $E=\sk^2$ if and only if $\sk\in \Sigma^>$. The multiplicity 
$n(\sk)$ is finite. The set $\Sigma^>$ is discrete and has no finite 
accumulation point in $\R_+$.
Any eigenfunction to a positive eigenvalue is identically zero 
on any external edge.
\end{theorem}
For special boundary conditions, one can obtain many positive eigenvalues, 
just take for example Dirichlet or Neumann boundary conditions everywhere. 
On the other hand, there are also nontrivial boundary conditions, that is ones which do not decouple 
the external edges from the internal ones, and which give positive 
eigenvalues, see Example 3.2 in \cite{KS1} and Example 4.3 in \cite{KS4}. Also there are examples 
with \emph{standard boundary conditions} (cf. e.g. Example 4.5 in \cite{KS8} for the definition), 
for which there are positive eigenvalues \cite{Kos}.
\begin{corollary}\label{cor:sigma}
The quantities $S(\sk),\alpha(\sk)$ and $\beta(\sk)$ depend smoothly on 
$\sk\in \R_+\setminus \Sigma^>$.
\end{corollary}
\begin{proof}
$Z_{A,B}(\sk)$ is analytic in $\sk\in\C$, so 
$Z_{A,B}(\sk)^{-1}$ is smooth in $\sk\in \R_+\setminus \Sigma^>$ and the 
claim follows from the representation \eqref{defs}.
\end{proof}
For further reference we denote by $\psi^{\sk,\nu}$ for $\sk\in\Sigma^>$ and 
$1\le \nu\le n(\sk)$ an orthonormal basis of the eigenspace with eigenvalue $E=\sk^2>0$.
By what has just been proved, each such eigenfunction is necessarily of the form
\begin{equation}\label{posevform}
\psi^{\sk,\nu}_{j}(x)=\begin{cases}0 & \text{for}\;j\in\cE \\
u^{\sk,\nu}_{j} \e^{\ii\sk x}+
v^{\sk,\nu}_{j} \e^{-\ii\sk x} & \text{for}\; j\in\cI. \end{cases}
\end{equation}
The orthonormality condition for fixed $\sk$ is obviously
\begin{align}\label{posevform1}
\langle\psi^{\sk,\nu},\psi^{\sk,\nu^\prime}\rangle&=\delta_{\nu,\nu^\prime}
=\sum_{i\in\cI}\Big\{u^{\sk,\nu}_{i}\overline{u^{\sk,\nu^\prime}_{i}}a_i+
v^{\sk,\nu}_{i}\overline{v^{\sk,\nu^\prime}_{i}}a_i \\\nonumber
&\qquad\qquad\qquad
+\frac{1}{2\ii\sk}\left(\overline{v^{\sk,\nu}_{i}}u^{\sk,\nu^\prime}_{i}
\left(\e^{2\ii \sk a_i}-1\right)
-\overline{u^{\sk,\nu}_{i}}v^{\sk,\nu^\prime}_{i}
\left(\e^{-2\ii \sk a_i}-1\right)\right)\Big\},
\end{align}
a quadratic form in the $u$'s and $v$'s. Thus we obtain
\begin{corollary}\label{cor:degen}
The degeneracy $n(\sk)$ of any discrete eigenvalue $E=\sk^2 >0$, that is $\sk\in\Sigma^>$,
satisfies the bound
\begin{equation}\label{degenbound}
n(\sk)\le 2|\cI|.
\end{equation}
In particular $\Sigma^>$ is empty when $\cG$ is a single vertex graph.
\end{corollary}
This result compares with Proposition \ref{prop:spectrum}.
We turn to a  
\underline{Proof of Proposition \ref{prop:indep}}.
Linear independence is clear due to the different occurrence of incoming waves in the different
$\psi^l(\,;\sk)$. Assume now that $\psi$ satisfies 
$-\Delta_{A,B}\psi=\sk^2\psi$ and the boundary conditions 
\eqref{bdycond}.The components are necessarily of the 
form $\psi_j(x)=u_j\e^{\ii \sk x}+v_j\e^{-\ii \sk x}$ for all $j\in\cE\cup\cI$. 
Set 
$\phi=\psi-\sum_{k\in\cE}v_k \psi^k(\,;\sk)$ such that $\phi$ also satisfies 
$-\Delta_{A,B}\phi=\sk^2\phi$ and the boundary conditions. We have to 
show that $\phi=0$. Observe that by construction the components are of the form
$$
\phi_j(x)=\begin{cases}\hat{s}_j e^{\ii \sk x},\quad& j\in\cE\\
\hat{u}_j e^{\ii \sk x} +\hat{v}_j\e^{-\ii \sk x},\quad& j\in\cI\end{cases}
$$
such that $\phi$ contains no incoming waves. Therefore the 
boundary conditions can be written in the form
\begin{equation}\label{posevbdy}
Z(\sk)\begin{pmatrix}\underline{s}(\sk)\\\underline{u}(\sk)
\\\underline{v}(\sk)\end{pmatrix}=0
\end{equation}
with 
$$
\underline{s}(\sk)=\{\hat{s}_k\}_{k\in\cE},\quad 
\underline{u}(\sk)=\{\hat{u}_j\}_{j\in\cI},
\quad \underline{v}(\sk)=\{\hat{v}_j\}_{j\in\cI}
$$
viewed as column vectors.
By assumption $\sk\notin\Sigma^>$, so $\hat{s}_k=\hat{u}_j=\hat{v}_j=0$ for all 
$k\in\cE, j\in\cI$, and $\phi$ indeed vanishes thus concluding the proof of Proposition \ref{prop:indep}. 

\begin{theorem}(\cite{KS1} Theorem 3.12, \cite{KS8} Corollary 3.16)
\label{theo:unitarity}
The scattering matrix is unitary for all $\sk>0$,
\begin{equation}\label{unitarity}
S(k)^\dagger=S(\sk)^{-1}.
\end{equation}
In addition the identity 
\begin{equation}\label{-ks}
S(-\sk)=S(\sk)^{-1}
\end{equation}
between meromorphic matrix valued functions in $\sk$ is valid. 
\end{theorem}
There are analogous relations for $\alpha(\sk),\beta(\sk)$ in the form
\begin{lemma}\label{lem:hermanal}
The following identities for meromorphic matrix valued functions in $\sk\in\C$ hold
\begin{align}\label{-k}
\alpha(-\sk)&=\beta(\sk)S(-\sk)\\\nonumber
\beta(-\sk)&=\alpha(\sk)S(-\sk).
\end{align}
\end{lemma}
\begin{proof}
We will simultaneously also give a new proof of \eqref{-ks}. Arrange the 
components $\psi^l_j(\,;\sk)$ 
as a $(|\cE|+|\cI|)\times |\cE|$
matrix $\psi(\,;\sk)$, such that the components of $\psi^l(\,;\sk)$ form the $l^{th}$ 
column. In view of \eqref{psidef}, the claims \eqref{-ks} and \eqref{-k} 
combined are equivalent to the relation 
\begin{equation}\label{alphabeta4}
\psi(\,;-\sk)=\psi(\,;\sk)S(-\sk)
\end{equation}
as an identity of meromorphic matrix valued functions.
Here, by the meromorphic properties of $S(\sk),\alpha(\sk)$ and $\beta(\sk)$, we view
each $\psi^l(\,;\sk)$ as meromorphic in $\sk$, that is each component $\psi^l_j(x;\sk)$ is 
meromorphic in $\sk$. So if we {\it define} 
\begin{equation}\label{alphabeta5}
\widehat{\psi}(\,;\sk)=\psi(\,;-\sk)S(\sk)
\end{equation}
we have to show that 
\begin{equation}\label{alphabeta6}
\widehat{\psi}(\,;\sk)=\psi(\,;\sk)
\end{equation}
holds as an identity between meromorphic matrix valued functions. Now  
$-\Delta_{A,B}\psi^l(\,;\sk)=\sk^2\psi^l(\,;\sk)$ holds. Moreover
the boundary values $\underline{\psi^l(\,;\sk)}$ and $\underline{\psi^l(\,;\sk)}^\prime$ of
$\psi^l(\,;\sk)$, see \eqref{bdvpsi}, are also meromorphic. Since the boundary conditions are 
satisfied for all $\sk>0$, they also hold for all $\sk$ away from the poles by the identity 
theorem for analytic functions. Therefore they also hold for all 
$\psi^l(\,;-\sk)$ and hence also for all $\widehat{\psi}^l(\,;\sk)$.
Similarly $-\Delta_{A,B}\psi^l(\,;\sk)=\sk^2\psi^l(\,;\sk)$ implies
$-\Delta_{A,B}\psi^l(\,;-\sk)=\sk^2\psi^l(\,;-\sk)$ and therefore also
$-\Delta_{A,B}\widehat{\psi}^l(\,;\sk)=\sk^2\widehat{\psi}^l(\,;\sk)$.
Again by the identity theorem for meromorphic functions it suffices to prove 
\eqref{alphabeta6} for all $\sk\in\R_+\setminus \Sigma^>$. 
But by Proposition \ref{prop:indep} each 
$\widehat{\psi}^l(\,;\sk)$ is a linear combination of the $\psi^k(\,;\sk)$.
By construction
\begin{equation}\label{hatpsidef}
\widehat{\psi}^{l}_{j}(x;\sk)=\begin{cases} 
  \e^{-\ii\sk x}\delta_{jl} + 
{S(\sk)}_{jl} \e^{\ii\sk x} & \text{for}\;j\in\cE \\
&\\
                       \left(\alpha(-\sk)S(\sk)\right)_{jl} \e^{-\ii\sk x}+
\left(\beta(-\sk)S(\sk)\right)_{jl} \e^{\ii\sk x} 
& \text{for}\; j\in\cI. \end{cases}
\end{equation}
But the eigenfunctions $\psi^l(\,;\sk)$ and $\widehat{\psi}^{l}(\,;\sk)$ satisfy the same defining 
properties and so by the uniqueness of $S(\sk),\alpha(\sk)$ and $\beta(\sk)$ we infer
\eqref{alphabeta6}.
\end{proof}
\begin{remark}\label{reality}
Since $S(\sk)$ is meromorphic in $\sk$, its unitarity for positive $\sk$ 
extends to complex $\sk$ in the form of hermitian analyticity \cite{ ELOP,Olive}
\begin{equation}\label{hermanal}
S(\sk)^\dagger=S(\bar{\sk})^{-1}.
\end{equation}
Combined with \eqref{-k} this gives 
\begin{equation}\label{unitarity1}
S(\sk)^\dagger=S(-\bar{\sk}).
\end{equation}
In particular $S(\sk)$ is a hermitian matrix when $\sk$ is purely 
imaginary. Since each $\psi^l(\,;-\sk)$ satisfies 
$-\Delta_{A,B}\psi^l(\,;-\sk)=\sk^2\psi^l(\,;-\sk)$
and the boundary conditions $(A,B)$, it has to be a linear combination of the 
$\psi^{l^\prime}(\,;\sk)$ and so \eqref{alphabeta4} just provides the explicit form.
\end{remark}
We consider the behavior under complex conjugation. 
Observe that if $(A,B)$ has maximal rank and $AB^\dagger$ is hermitian, 
then the complex conjugate pair $(\bar{A},\bar{B})$ is also of maximal rank and $\bar{A}\bar{B}^\dagger$ 
is hermitian. So $(\bar{A},\bar{B})$ also gives rise to a Laplacian. 
The following lemma is trivial.
\begin{lemma}\cite{KS1}\label{lem:bdycomplex}
If $\psi$ satisfies the boundary condition $(A,B)$ then the complex conjugate 
wave function $\bar{\psi}$ satisfies the boundary condition $(\bar{A},\bar{B})$.

In particular, if $\psi$ is in the domain of $\Delta_{A,B}$, then 
$\bar{\psi}$ is in the domain of $\Delta_{\bar{A},\bar{B},\underline{a}}$ and
\begin{equation}\label{complexconj}
\overline{-\Delta_{A,B}\psi}
=-\Delta_{\bar{A},\bar{B},\underline{a}}\bar{\psi}
\end{equation}
holds. 
\end{lemma}
This gives the following nice observation, whose proof we omit. Recall relation 
\eqref{nplusbar} in connection with Proposition \ref{prop:spectrum}.
\begin{corollary}\label{conjspectrum}
The spectra of the two Laplacians $\Delta_{A,B}$ and $\Delta_{\bar{A},\bar{B},\underline{a}}$ 
agree. Moreover, if $\psi$ is an (improper) eigenfunction of $-\Delta_{A,B}$, then 
$\bar{\psi}$ is an (improper) eigenfunction of $\Delta_{\bar{A},\bar{B},\underline{a}}$ for the 
same eigenvalue.
\end{corollary}
Let $^T$ denote transposition of a matrix. 
\begin{lemma}(\cite{KS4} Theorem 2.2)\label{lem:alphabeta}
The following identities between meromorphic matrix valued functions 
hold for arbitrary boundary conditions $(A,B)$ 
\begin{align}\label{alphabeta7}
S_{\bar{A},\bar{B}}(\sk)&=S_{A,B}(\sk)^T\\\nonumber
\alpha_{\bar{A},\bar{B}}(\sk)&
=\overline{\beta_{A,B}(\bar{\sk})}
\;S_{A,B}(\sk)^T\\\nonumber
\beta_{\bar{A},\bar{B}}(\sk)&
=\overline{\alpha_{A,B}(\bar{\sk})}
\;S_{A,B}(\sk)^T.
\end{align}
or equivalently
\begin{equation}\label{alphabeta71}
\psi_{\bar{A},\bar{B}}(\,;\sk)=\overline{\psi_{A,B}(\,;\bar{\sk})}S_{A,B}(\sk)^T.
\end{equation}
\end{lemma}

\begin{proof}
We give an alternative proof along the lines used in the proof 
of Lemma \ref{lem:hermanal}. Indeed, with the notation used there, define
for complex $\sk$
\begin{equation}\label{alphabeta8}
\check{\psi}(\,;\sk)=\overline{\psi_{A,B}(\,;\bar{\sk})}\:\:
\overline{S_{A,B}(\bar{\sk})^{-1}}
\end{equation}
where we indicate the dependence on the boundary conditions. The aim is to show
\begin{equation}\label{alphabeta9}
\check{\psi}(\,;\sk)=\psi_{\bar{A},\bar{B}}(\,;\sk),
\end{equation}
from which \eqref{alphabeta7} and \eqref{alphabeta71} follow.
Again by the identity theorem for meromorphic functions it suffices to prove this 
relation for $\sk>0$, for which $\sk^2$ is not a discrete eigenvalue. For such $\sk$
by unitarity $\overline{S_{A,B}(\bar{\sk})^{-1}}=S_{A,B}(\sk)^T$ and hence for all 
$\sk\in\C$, again by the identity theorem. By Lemma \eqref{lem:bdycomplex} each 
$\check{\psi}^l(\,;\sk)$ is an improper eigenfunction of 
$-\Delta(\bar{A},\bar{B},\underline{a})$ with eigenvalue $\sk^2$ and hence must be a 
linear combination of the $\psi_{\bar{A},\bar{B}}^k(\,;\sk)$.
By construction, the components of 
$\check{\psi}^l(\,;\sk),\, \sk>0$ are of the form
\begin{align}\label{checkpsidef}
\check{\psi}^{l}_{j}(x;\sk)&=\begin{cases} 
  \e^{-\ii\sk x}\delta_{jl} + 
  \overline{S_{A,B}(\sk)^{-1}}_{jl}\e^{\ii\sk x} & \text{for}\;\;j\in\cE \\
&\\
  \overline{\left(\alpha_{A,B}(\sk)S_{A,B}(\sk)^{-1}\right)}_{jl} 
\e^{\ii\sk x}+
\overline{\left(\beta_{A,B}(\sk)S_{A,B}(\sk)^{-1}\right)}_{jl} \e^{-\ii\sk x} 
& \text{for}\;\; j\in\cI \end{cases}\\\nonumber
&\\\nonumber
&=\begin{cases} 
  \e^{-\ii\sk x}\delta_{jl} + S_{A,B}(\sk)_{lj}\e^{\ii\sk x} & \text{for}\;j\in\cE \\
&\\
\left(\overline{\alpha_{A,B}(\sk)}S_{A,B}(\sk)^T\right)_{jl} 
\e^{\ii\sk x}+
\left(\overline{\beta_{A,B}(\sk)}S_{A,B}(\sk)^T\right)_{jl} \e^{-\ii\sk x} 
& \text{for}\; j\in\cI. \end{cases}\nonumber
\end{align}
But $\psi^l_{\bar{A},\bar{B}}(\,;\sk)$ and 
$\check{\psi}^{l}(\,;\sk)$ satisfy the same defining 
properties and so by the uniqueness of 
$S_{\bar{A},\bar{B}}(\sk),\alpha_{\bar{A},\bar{B}}(\sk)$ and 
$\beta_{\bar{A},\bar{B}}(\sk)$ we infer \eqref{alphabeta9}.
\end{proof}
By Corollary \eqref{conjspectrum} we know that $\overline{\psi^l(\,;\sk)}=\overline{\psi^l_{A,B}(\,;\sk)}$ 
are eigenfunctions of $-\Delta_{\bar{A},\bar{B},\underline{a}}$ with eigenvalue $\sk^2$. Relation 
\eqref{alphabeta71} tells us that they span the eigenspace of $-\Delta_{\bar{A},\bar{B}}$ 
for that eigenvalue since the $\psi^l_{\bar{A},{B}}(\,;\sk)$ do. We shall make use of this observation
when we construct massive, free charged fields in Section \ref{subsec:complex}.

By definition the boundary conditions given by the pair $(A,B)$ are {\it real}
if an invertible $C$ exists such that the pair $(A^\prime,B^\prime)=(CA,CB)$ 
consists of real matrices $A^\prime$ and $B^\prime$. An equivalent condition is that there exists an 
invertible $C^\prime$ with $C^\prime A=\bar{A}$ and $C^\prime B=\bar{B}$, see \cite{KS4}. 
As a direct consequence of Lemmas \ref{lem:hermanal} and \ref{lem:alphabeta} 
we obtain the following two corollaries 
\begin{corollary}\label{cor:complex}
For arbitrary boundary conditions $(A,B)$, the relations
\begin{equation}\label{kcomplex}
\overline{\alpha_{A,B}(\bar{\sk})}=\alpha_{\bar{A},\bar{B}}(-\sk),\qquad 
\overline{\beta_{A,B}(\bar{\sk})}=\beta_{\bar{A},\bar{B}}(-\sk)
\end{equation}
hold as identities between matrix valued meromorphic functions in $\sk\in\C$.
\end{corollary}
\begin{corollary}\label{cor:real}
If the boundary conditions $(A,B)$ are real, then the relations
\begin{equation}\label{-kreal1}
\overline{S(\bar{\sk})}=S(-\sk),\quad
\overline{\beta(\bar{\sk})}=\alpha(\sk)S(-\sk),\quad
\overline{\alpha(\bar{\sk})}=\beta(\sk)S(-\sk)
\end{equation}
and hence 
\begin{equation}\label{-kreal2} 
\overline{\alpha(\bar{\sk})}=\alpha(-\sk),\qquad
\overline{\beta(\bar{\sk})}=\beta(-\sk)
\end{equation}
are valid as identities between matrix valued meromorphic functions in $\sk\in\C$.
\end{corollary} 
As a consequence of Lemma \ref{lem:bdycomplex} we directly obtain 
\begin{corollary}\label{cor:laplacereal2}
For real boundary conditions $(A,B)$ $\bar{\psi}$ is an eigenfunction of 
$-\Delta_{A,B}$ whenever  $\psi$ is. Therefore for a given eigenvalue, the associated eigenspace is 
spanned by real eigenfunctions.
\end{corollary}
So if for real boundary conditions we choose the eigenfunctions $\psi^{\sk,\nu}$ 
to be real, then in the notation of
\eqref{posevform} the relations 
\begin{equation}\label{realcoeff}
u_j^{\sk,\nu}=\overline{v_j^{\sk,\nu}},\qquad\qquad\qquad \qquad\qquad\qquad
\quad\sk\in\Sigma^>
\end{equation}
are valid. Similarly, we can rewrite \eqref{-kreal1} as
\begin{corollary}\label{cor:laplacereal3}
If the boundary conditions are real, then the relation 
\begin{equation}\label{realimprop}
\overline{\psi(\,;\bar{\sk})}=\psi(\,;\sk)S(-\sk)
\end{equation}
is valid. 
\end{corollary}
Also \eqref{hermanal} and the first relation in \eqref{-kreal1} gives
\begin{lemma}(see \cite{KS1} Corollary 3.2, \cite{KS4} Theorem 2.2)\label{lem:reales}
If the boundary conditions $(A,B)$ are real, then $S(\sk)$ is 
a symmetric matrix and so for $\sk$ purely imaginary the matrix $S(\sk)$ is real due 
to \eqref{unitarity1}.
\end{lemma}
\begin{remark}\label{re:timereversal}
For arbitrary boundary conditions $(A,B)$ the equivalent exponentiated form of \eqref{complexconj} is 
\begin{equation}\label{timereversal}
\overline{e^{\ii\Delta_{A,B}t}\,\psi}
=\e^{-\ii\Delta_{\bar{A},\bar{B},\underline{a}}\,t}\,\overline{\psi}.
\end{equation}
If the boundary conditions $(A,B)$ are real and hence 
$\Delta_{A,B}=\Delta_{\bar{A},\bar{B},\underline{a}}$ holds, 
then \eqref{timereversal} is just the statement that {\rm time reversal invariance} holds. 
In the single vertex case this invariance combined with the hermiticity condition on the field (see below) 
has been used in \cite{BeMiSoI} to prove that $S(\sk)$ is then a symmetric matrix .
\end{remark}
Combined with \eqref{kindeps} we obtain
\begin{corollary}\label{cor:laplacereal4} 
For a single vertex graph all $\sk$-independent S-matrices resulting from {\rm real} boundary conditions 
are of the form \eqref{kindeps} where $P$ is a real, symmetric and idempotent matrix, $P^2=P$.
\end{corollary}

\subsection{Negative eigenvalues of the Laplace operator and their eigenfunctions}
\label{sec:negev}
The operator $-\Delta_{A,B}$ may have negative eigenvalues.
We introduce the sets 
\begin{align}\label{sigma<}
\Sigma^{\le}&=\Sigma^<_{A,B}=\{\sk=\ii \kappa\mid\: \kappa\ge 0, \sk^2=-\kappa^2\quad 
\mbox{is an eigenvalue of}\quad -\Delta_{A,B}\}\\\nonumber
\Sigma^<&=\Sigma^<_{A,B}=\{\sk=\ii \kappa\mid\: \kappa>0, \sk^2=-\kappa^2\quad 
\mbox{is an eigenvalue of}\quad -\Delta_{A,B}\}
\end{align}
such that trivially $\Sigma^<\subseteq\Sigma^{\le}$ and let $\Sigma=\Sigma^{\le}\cup\Sigma^>$, the set of all
discrete eigenvalues. 
We will discuss zero as a possible eigenvalue 
separately in the next subsection \ref{subsec:zero}.
Since all Laplace operators $-\Delta_{A,B}$ for different $(A,B)$ 
are finite rank perturbations of each other and since 
the ones with Dirichlet and (or) Neumann boundary conditions are non-negative, 
$\Sigma^<$ is a finite set and the multiplicity of each eigenvalue is finite.
If $\sk^2=-\kappa^2<0$ is such an eigenvalue with multiplicity 
$n(\sk)$, there is a finite, orthonormal basis of eigenfunctions 
$\psi^{\sk,\nu},\;1\le \nu\le n(\sk)$. Written in local coordinates they are all necessarily of the form
\begin{equation}\label{ef}
\psi^{\sk,\nu}_j(x)=\begin{cases}s_j^{\sk,\nu}\e^{\ii\sk x}
\quad&\mbox{for}\quad j\in\cE\\
u_j^{\sk,\nu}\e^{\ii \sk x}+v_j^{\sk,\nu}
\e^{-\ii\sk x}\quad&\mbox{for}\quad j\in\cI.
\end{cases}
\end{equation}
The orthonormality condition for fixed $\sk\in\Sigma^<$ is easily calculated to be 
\begin{align}\label{efnormal}
\delta_{\nu,\nu^\prime}=\langle\psi^{\sk,\nu},\psi^{\sk,\nu^\prime}\rangle&=
-\frac{1}{2\ii\sk}\sum_{e\in\cE}\overline{s_j^{\sk,\nu}}s_j^{\sk,\nu^\prime}
+\sum_{i\in\cI}\Big\{u^{\sk,\nu}_{i}\overline{u^{\sk,\nu^\prime}_{i}}a_i+
v^{\sk,\nu}_{i}\overline{v^{\sk,\nu^\prime}_{i}}a_i \\\nonumber
&\qquad+\frac{1}{2\ii\sk}\left(\overline{v^{\sk,\nu}_{i}}u^{\sk,\nu^\prime}_{i}
\left(\e^{2\ii \sk a_i}-1\right)
-\overline{u^{\sk,\nu}_{i}}v^{\sk,\nu^\prime}_{i}
\left(\e^{-2\ii \sk a_i}-1\right)\right)\Big\}.
\end{align}
In analogy to Corollary \ref{cor:degen} we obtain 
\begin{corollary}\label{degen-}
The degeneracy $n(\sk)$ of any discrete eigenvalue $E=\sk^2 \;(\sk\in\Sigma^<)$ 
satisfies the bound
\begin{equation}\label{degenbound-}
n(\sk)\le |\cE|+2|\cI|.
\end{equation}
\end{corollary}
After a short calculation, the boundary condition can be brought into the form, 
compare 
\eqref{posevbdy},
\begin{equation}\label{efbdy}
Z(\sk=\ii\kappa)\begin{pmatrix}\underline{s}^{\sk=\ii\kappa,\nu}
\\\underline{u}^{\sk=\ii\kappa,\nu}
\\\underline{v}^{\sk=\ii\kappa,\nu}\end{pmatrix}=0.
\end{equation}
In case the boundary conditions are real, the $\psi^{\sk,\nu}$ may be chosen 
to be real, that is the coefficients $s_e^{\sk,\nu},u_j^{\sk,\nu}$ and $v_j^{\sk,\nu}$ are all real. 

Recall that there is a canonical Lebesgue measure $dp$ on $\cG$.   
$\delta(p,q)$ is the Dirac $\delta$-function on $\cG$ with the defining property
$$
\int_\cG \delta(p,q)f(q)dq=f(p). 
$$
\begin{remark}\label{re:completeness}
The arguments may also be reversed to show that $\Sigma$ equals 
the set of zeros of $\det Z(\sk)$ in the set $\{\sk\in\C\mid \Re \sk=0,\Im\, \sk>0\}\cup \R_+$ 
and that the $\sk^2$ with $\sk\in\Sigma$ form exactly the discrete spectrum.
As a result there is a completeness relation written as 
\begin{equation}\label{compl}
\frac{1}{2\pi}\sum_l\int_0^\infty d\sk\; \psi^l(p;\sk)\overline{\psi^l(q;\sk)}
+\sum_{\sk\in\Sigma,1\le \nu\le n(\sk)}\psi^{\sk,\nu}(p)
\overline{\psi^{\sk,\nu}(q)}
=\delta(p,q)\qquad p,q\in\cG.
\end{equation}
The normalization factor $1/2\pi$ is due to \eqref{ortho1}.
\end{remark}
\subsubsection{Bound states and poles of the S matrix in the single vertex case}\label{boundstatespoles}~~~\\
In the single vertex case one can actually say much more about bound states. In fact, we will see that 
they are completely encoded in the S-matrix.
Thus the negative eigenvalues are the poles of the scattering matrix and the corresponding 
eigenfunctions may be obtained from the residues of the poles. 
To explain this in detail let $\cG_n$ denote the single vertex graph with $n=|\cE|$ 
half-lines meeting at the single vertex $v$. 
We will label these half-lines from $1$ to $n$. 
The scattering matrix $S(\sk)$ now simply equals
$\mathfrak{S}(\sk)=-(A+\ii \sk B)^{-1}(A-\ii \sk B)$. As shown in \cite{KS8}, see 
relation (3.23) there, the $S(\sk)$ for 
different $\sk$ all commute and as a consequence there is a common spectral 
decomposition \cite{KPS3}, 
\begin{equation}\label{star1}
S(\sk)=\sum_{\kappa\in\fI}S^\kappa(\sk)= \sum_{\kappa\in\fI} \frac{\sk+\ii\kappa}{\sk-\ii\kappa}P^\kappa=
P^0-P^\infty+\sum_{\kappa\in\fI_0} \frac{\sk+\ii\kappa}{\sk-\ii\kappa}P^\kappa.
\end{equation}
$\fI=\fI(A,B)$ is a finite set of different real numbers, including possibly 
the values $\kappa=0,\infty$. Also $\fI_0$ is the subset, where these elements have been omitted. 
The $P$'s define a 
decomposition of unity of pairwise orthogonal projectors
\begin{equation}\label{star2}
\sum_{\kappa\in \fI}P^\kappa=\1_{n\times n}, \quad 
(P^\kappa)^\dagger=P^\kappa,\quad  P^\kappa P^{\kappa^\prime}
=P^\kappa\delta_{\kappa\kappa^\prime}. 
\end{equation}
Thus $P^\kappa$ is the orthogonal projection onto the eigenspace of $S(\sk)$ 
with eigenvalue $(\sk+\ii\kappa)/(\sk-\ii\kappa)$ (equal to $1$ for $\kappa=0$ and equal to 
$-1$ for $\kappa=\infty$). The multiplicities are $0\le n_S(\ii\kappa)=\tr P^\kappa$.

$S(\sk)$ is $\sk$-independent if and only if 
$\fI=\{0,\infty\}$, that is $\fI_0=\emptyset$. Then $P$ in \eqref{kindeps} is just $P^\infty$ and  
$P^0=\1-P^\infty$. Moreover $S(\sk)$ is invertible if and only if $\sk\notin \ii\fI_0\cup-\ii\fI_0$. 
\begin{lemma}\label{lem:realproj}
If the boundary conditions $(A,B)$ are real, then the $P^\kappa$ are 
real, symmetric matrices.
\end{lemma}
\begin{proof}
Due to the representation
\begin{equation*}\label{prep1}
P^\kappa=\lim_{\tau\rightarrow \kappa} \frac{\tau-\kappa}{\tau+\kappa}S(\ii\tau),\quad \kappa\neq 0,\infty,
\end{equation*}
each $P^\kappa$ with $\kappa\in\fI_0$ is real and symmetric by Lemma \ref{lem:reales}. So 
by the same lemma and with $\tau>\max_{\kappa\in\fI} \kappa$
\begin{equation*}\label{prep2}
\begin{aligned}
P^{0}-P^{\infty}&=S(\ii \tau)-\sum_{\kappa\in\fI_0 } \frac{\tau+\kappa}{\tau-\kappa}P^\kappa\\
P^{0}+P^{\infty}&=\1_{n\times n}-\sum_{\kappa\in\fI_0} P^\kappa
\end{aligned}
\end{equation*}
are real and symmetric and so are both $ P^{0}$ and $P^{\infty}$.
\end{proof}
Our next aim is to determine the eigenfunctions $\psi^{\sk,\nu}$ out of these data.
We will conform to our previous notation and show 
$\Sigma=\Sigma^<=\{\ii\kappa\mid 0<\kappa\in\fI_0 \}$, such that the 
negative eigenvalues of $-\Delta_{A,B}$ are of the form $-\kappa^2$. Given $\kappa$, 
there are orthonormal unit vectors 
$\underline{s}^{\ii\kappa, \nu}\, (1\le\nu\le n(\ii\kappa))$ in $\C^n$ which span the 
eigenspace of $P^\kappa$ for the eigenvalue $1\;(=\Ran P^\kappa )$, 
\begin{equation}\label{sortho}
P^{\kappa^\prime}\underline{s}^{\ii\kappa, \nu}=\delta_{\kappa^\prime,\kappa}\underline{s}^{\ii\kappa, \nu}
\end{equation}
and hence
\begin{equation}\label{seigen}
S(\sk)\underline{s}^{\ii\kappa, \nu}=\frac{\sk+\ii\kappa}{\sk-\ii\kappa}\underline{s}^{\ii\kappa, \nu}.
\end{equation} 
Observe that the entire set of the $\underline{s}^{\ii\kappa, \nu}$ is automatically orthonormal
by \eqref{star2} and \eqref{sortho}
\begin{equation}\label{snorm}
\sum_{j=1}^n\overline{s^{\ii\kappa, \nu}_j}s^{\ii\kappa^\prime, \nu^\prime}_j
=\delta_{\kappa,\kappa^\prime}\delta_{\nu,\nu^\prime}.
\end{equation}
When the boundary conditions and hence also the projectors are real by the previous 
lemma, these eigenvectors may then be chosen to be real. 
We define the family of functions 
$\psi^{\ii\kappa, \nu}$ in $L^2(\cG_n)$ in terms of its components as
\begin{equation}\label{psimunu}
\psi^{\ii\kappa, \nu}_j(x)=s^{\ii\kappa, \nu}_j\sqrt{2\kappa}\e^{-\kappa x}\qquad 1\le j\le n, 
1\le \nu\le n(\ii\kappa).
\end{equation}
The orthonormality of this set follows from the orthonormality \eqref{snorm} of the 
$\underline{s}^{\ii\kappa, \nu}$. 
The next result shows how the bound states are encoded in the scattering matrix. 
\begin{proposition}\label{prop:starev}
Let $\cG$ be a single vertex graph. For $0<\kappa<\infty$ appearing in the spectral 
decomposition \eqref{star1} the $\psi^{\ii\kappa, \nu}$ as defined by \eqref{psimunu} are normalized 
eigenfunctions of $-\Delta_{A,B}$ 
with eigenvalue $-\kappa^2$ satisfying the boundary conditions. Conversely, if $-\kappa^2$ 
with $0<\kappa <\infty$ is an eigenvalue, then $S(\sk)$ has a pole at $\sk=\ii \kappa$.
In particular the multiplicity of each such eigenvalue is $n_S(\ii\kappa)$ and the number $n_b$ 
of bound states (counting multiplicities) equals
$$
n_b=\sum_{0<\kappa \in\fI_0}n_S(\ii\kappa).
$$
\end{proposition}
By this lemma there are at most $n=|\cE|$ bound states when $\cG$ is a single vertex graph. It is easy 
to construct examples where this upper bound actually is also obtained. 
\begin{proof}
As for the boundary values of $\psi^{\ii\kappa, \nu}$ and its derivative we have 
$$
\underline{\psi}^{\ii\kappa, \nu}=\sqrt{2\kappa}\,\underline{s}^{\ii\kappa, \nu},\quad 
\underline{\psi}^{\ii\kappa,\nu\;\prime}
=-\kappa \sqrt{2\kappa}\,\underline{s}^{\ii\kappa, \nu}
=-\kappa\, \underline{\psi}^{\ii\kappa, \nu}
$$
and we have to show that 
$$
A\underline{\psi}^{\ii\kappa, \nu}+B{\underline{\psi}^{\ii\kappa, \nu}}^\prime=
\sqrt{2\kappa}\,(A-\kappa B)\underline{s}^{\ii\kappa, \nu}=0.
$$
As established in \cite{KS2}, see also \cite{KS8} Proposition 3.7, we may instead of 
$(A,B)$ equivalently use the pair $(A(\sk),B(\sk))$, where 
$$
A(\sk)=-\frac{1}{2}\left(S(\sk)-\1_{n\times n}\right),\qquad  B(\sk)
=\frac{1}{2\ii\sk}
\left(S(\sk)+\1_{n\times n}\right)
$$
and where $\sk>0$ is arbitrary. Actually by the proof given there, 
$\sk$ may be chosen arbitrary in the domain of analyticity of 
$S(\sk)$ and for which $S(\sk)$ is invertible. By 
\begin{equation}\label{sdet}
\det S(\sk)=(-1)^{n_S(\ii\infty)}
\prod_{\kappa\in\fI_0}\left(\frac{\sk+\ii\kappa}{\sk-\ii\kappa}\right)^{n_S(\ii\kappa)}
\end{equation}
this is the case if 
$\sk$ is chosen outside the set $\ii\fI_0\cup-\ii\fI_0$. In a moment we shall have occasion 
to make use of this observation. A trivial calculation using \eqref{seigen} gives  
$$
\begin{aligned}
(A(\sk)-\kappa B(\sk))\underline{s}^{\kappa, \nu}&=
\frac{1}{2}\left(\left(-1-\frac{\kappa}{\ii\sk}\right)S(\sk)+
\left(1-\frac{\kappa}{\ii\sk} \right)\1_{n\times n}\right)\underline{s}^{\kappa, \nu}\\
&=\frac{1}{2}\left(\left(-1-\frac{\kappa}{\ii\sk}\right)\frac{\sk+\ii \kappa}
{\sk-\ii \kappa} + \left(1-\frac{\kappa}{\ii\sk} \right)\right)
\underline{s}^{\kappa, \nu}=0.
\end{aligned}
$$
As for the converse let $\psi\neq 0$ be an eigenfunction of $-\Delta_{A,B}\psi$ with eigenvalue 
$-\kappa_0^2$ with $\kappa_0>0$, $-\Delta_{A,B}\psi=-\kappa^2_0\psi$. Then $\psi$ is
necessarily of the form $\psi_j(x)=c_j\exp-\kappa_0 x$. Let $0\neq \underline{c}\in\C^n$ denote the column 
vector with components $c_j$. Since $\psi$ satisfies the boundary conditions, the relation   
\begin{equation}\label{kappa0}
(A-\kappa_0 B)\underline{c}=0
\end{equation}
holds or equivalently by the above remarks 
$$
(A(\sk)-\kappa_0 B(\sk))\underline{c}=0,\qquad \sk\notin\ii\fI\cup-\ii\fI
$$
which when written out gives
\begin{equation}\label{spole}
S(\sk)\underline{c}=\frac{\sk+\ii\kappa_0}{\sk-\ii\kappa_0}\underline{c}.
\end{equation}
Thus $S(\sk)$ has a pole at $\sk=\ii\kappa_0\in\ii\fI$ and $P^{\kappa_0}\underline{c}=\underline{c}$
For a single vertex graph $Z(\sk)$ as defined by \eqref{zdef} equals $A+\ii\sk B$. 
Observe that we used only \eqref{kappa0} to establish \eqref{spole}.
Therefore the condition $\det (A+\ii\sk_0 B)=0$ for $\sk_0\neq 0$, cf. Remark \ref{re:completeness}, 
is equivalent to $\sk_0$ being a pole for $S(\sk)$. 
Moreover the discrete spectrum with its multiplicities 
is given in terms of the scattering matrix as
\begin{equation}\label{sigma1}
\Sigma=\{\ii\kappa\,|\,0<\kappa\in\fI_0\},\qquad n(\ii\kappa)=n_S(\ii\kappa),\quad 
0<\kappa\in\fI_0.
\end{equation}
\end{proof}
As a further, related consequence of this proposition the relation 
\begin{equation}\label{diracprep}
P^\kappa_{jl}=\sum_{1\le \nu\le n(\ii\kappa)} s_j^{\ii\kappa,\nu}\overline{s_l^{\ii\kappa,\nu}}
\end{equation}
holds for the matrix elements of $P^\kappa$.

As for the r\^ole of $P^0$ and $P^\infty$ we have 
\begin{lemma}\label{lem:pp}
The relations 
\begin{equation}\label{pp}
\ker A=\Ran P^0,\qquad \ker B=\Ran P^\infty
\end{equation}
hold, so in particular $AP^0=0,\,BP^\infty=0$.
\end{lemma}
\begin{proof}
By our previous discussion $\ker A=\ker A(\sk)=\ker (S(\sk)-\1)$ and 
$\ker B=\ker B(\sk)=\ker (S(\sk)+\1)$ for $\sk\notin\ii\fI\cup-\ii\fI$.
\end{proof}
The known relation $\ker A \perp \ker B=0$, see Lemma 3.4 in \cite{KS8}, is of course compatible 
with this result. Consider any piecewise constant function $\psi$, 
that is a function which is constant on each edge. 
Then $\psi$ is completely determined by its boundary values $\underline{\psi}$. Moreover, if 
$\underline{\psi}\in\ker A$, then $\psi$ satisfies the boundary condition \eqref{bdycond}.

\subsection{Zero as an eigenvalue}\label{subsec:zero}
In this subsection we establish necessary and sufficient conditions for $-\Delta_{A,B}$ to 
have $0$ as an eigenvalue. Let $\psi\neq 0$ be such a square integrable eigenfunction, $-\Delta_{A,B}\psi=0$.
Then necessarily $\psi_e(x)=0$ for all $e\in\cE$ while $\psi_i(x)=\gamma_i+\delta_i x$ for $i\in\cI$ and 
some $\gamma_i,\delta_i\in\C$, not all vanishing.
So with $\underline{\gamma}=\{\gamma_i\}_{i\in\cI},\underline{\delta}=\{\gamma_i\}_{i\in\cI}\in\C^{|\cI|}$, 
viewed as column vectors and with the notation \eqref{bdvpsi}
$$
\underline{\psi} = \begin{pmatrix} \underline{0}\\
                                   \underline{\gamma} \\
                                   \underline{\gamma}+T_{\underline{a}}\underline{\delta} \\
                                     \end{pmatrix}
                  =U_{\underline{a}}\begin{pmatrix}\underline{\gamma}\\\underline{\delta}\end{pmatrix},\quad
\underline{\psi}' = \begin{pmatrix} \;\;\underline{0}\\
                                   \;\;\underline{\delta}\\
                                   -\underline{\delta}\\
                                       \end{pmatrix}=
V\begin{pmatrix}\underline{\gamma}\\\underline{\delta}\end{pmatrix}
$$
with the diagonal matrix $T_{\underline{a}}=\diag\{a_i\}_{i\in\cI}$ and 
$$
U_{\underline{a}}=\begin{pmatrix}0&0\\\1&0\\\1&T_{\underline{a}}\end{pmatrix},\qquad 
V=\begin{pmatrix}0&\;0\\0&\;\1\\0&-\1\end{pmatrix}.
$$
So the boundary condition \eqref{bdycondrewr} takes the form 
\begin{equation}\label{bdyzeroev}
(A,B)\begin{pmatrix}\underline{\psi}\\\underline{\psi}'\end{pmatrix}=(AU_{\underline{a}},BV)
\begin{pmatrix}\underline{\gamma}\\\underline{\delta}\end{pmatrix}=0.
\end{equation}

Using the condition \eqref{bdycondrewr1} and viewing 
$$
\begin{pmatrix}U_{\underline{a}}\\V\end{pmatrix}
$$
as a $2(|\cE|+2|\cI|)\times 2|\cI|$ matrix 
we arrive at the 
\begin{proposition}\label{prop:zeroev}
The following relation is valid  
\begin{equation}\label{dimev0}
\dim \Ker \Delta_{A,B}= \dim \left(\Ker (A,B)\cap 
\Ran\begin{pmatrix}U_{\underline{a}}\\V \end{pmatrix}\right)=\dim\Ker (AU_{\underline{a}},BV).
\end{equation}
\end{proposition}
Observe that $\dim \Ker (A,B)=|\cE|+2|\cI|$ while 
$\dim\Ran\begin{pmatrix}U_{\underline{a}}\\V \end{pmatrix}=2|\cI|$ since 
$\Ker\begin{pmatrix}U_{\underline{a}}\\V \end{pmatrix}=0$.
Generically subspaces of these dimensions have trivial intersection in a space of dimension equal 
to $2(|\cE|+2|\cI|)$, that is they are transversal . This property 
remains valid even if one space, namely $\Ker (A,B)$, is required to be maximal isotropic.
As a consequence, for generic boundary conditions $(A,B)$ we conclude that 
$-\Delta_{A,B}$ does not have zero as an eigenvalue.
\begin{example}
Consider the interval $[0,a]$ with Robin boundary conditions at both ends
\begin{equation}\label{intervalrobin}
\cos\tau_0\psi(0)+\sin\tau_0\psi^\prime(0)=0,\qquad\cos\tau_1\psi(a)-\sin\tau_1\psi^\prime(a)=0.
\end{equation}
Then 
$$
(AU_{\underline{a}},BV)=\begin{pmatrix}\cos\tau_0&\sin\tau_0\\\cos\tau_1&a\cos\tau_1-\sin\tau_1 \end{pmatrix}
$$
has non-trivial kernel if and only if $a-\tan\tau_0-\tan\tau_1=0$ or $\cos\tau_0=\cos\tau_1=0$
( Neumann boundary conditions).
\end{example}

\vspace{0.3cm}
\subsection{Walk representation of the amplitudes $S(\sk),\alpha(\sk)$ and $\beta(\sk)$}
In this section we will provide an expansion of the amplitudes 
$S(\sk),\alpha(\sk)$ and $\beta(\sk)$ in terms of walks on the graph. We will use this result to give the 
proof of Theorem \ref{theo:mfinite} in Appendix \ref{app:KGsupport}. For the scattering matrix 
$S(\sk)$ such an expansion was already established in \cite{KS8}. The extension 
to $\alpha(\sk)$ and $\beta(\sk)$ is similar and goes as follows. For the 
convenience of the reader we recall those parts of the  notion of a walk as 
introduced in \cite{KS8} and extended in \cite{KS9} and which are relevant for our 
purpose.
A nontrivial walk $\bw$
on the graph $\cG$ from $j^\prime\in\cE\cup\cI$ to $j\in\cE\cup\cI$ is an
ordered sequence formed out of edges and vertices
\begin{equation}\label{walk:def}
\{j,v_0,j_1,v_1,\ldots,j_n,v_n,j^\prime\}
\end{equation}
such that
\begin{itemize}
\item[(i)]{$j_1,\ldots,j_n\in\cI$;}
\item[(ii)]{the vertices $v_0\in V$ and $v_n\in V$ satisfy $v_0\in\partial(j)$,
$v_0\in\partial(j_1)$, $v_n\in\partial(j^\prime)$, and
$v_n\in\partial(j_n)$;}
\item[(iii)]{for any $k\in\{1,\ldots,n-1\}$ the vertex $v_k\in V$ satisfies 
$v_k\in\partial(j_k)$ and $v_k\in\partial(j_{k+1})$;}
\item[(iv)]{$v_k=v_{k+1}$ for some $k\in\{0,\ldots,n-1\}$ if and only if $j_k$ is a 
tadpole.}
\end{itemize}
When $j,j^\prime\in\cE$ this definition is equivalent to that given in
\cite{KS8}.

The number $n$ is the \emph{combinatorial length} $|\bw|_{\mathrm{comb}}$
and the number
\begin{equation*}
|\bw|=\sum_{k=1}^n a_{j_k} >0
\end{equation*}
is the \emph{metric length} of the walk $\bw$.

A \emph{trivial} walk on the graph $\cG$ from $j^\prime\in\cE\cup\cI$ to
$j\in\cE\cup\cI$ is a triple $\{j,v,j^\prime\}$ such that
$v\in\partial(j)$ and $v\in\partial(j^\prime)$. Otherwise the walk is
called nontrivial. In particular, if $\partial(j)=\{v_0,v_1\}$, then
$\{j,v_0,j\}$ and $\{j,v_1,j\}$ are trivial walks, whereas
$\{j,v_0,j,v_1,j\}$ and $\{j,v_1,j,v_0,j\}$ are nontrivial walks of
combinatorial length $1$ and of metric length $a_j$. Both the combinatorial
and metric length of a trivial walk are zero.

We will say that the walk \eqref{walk:def} enters the final edge $j$ through the final 
vertex $v_0=v_0(\bw)$ and leaves the initial edge $j^\prime$ through the initial vertex 
$v_n=v_n(\bw)$. A trivial walk $\{j,v,j^\prime\}$ enters $j$ and leaves 
$j^\prime$ through the same vertex $v$.
Assume that the edges $j,j^\prime\in\cE\cup\cI$ are not tadpoles.
The following distance relation holds for a point $p$ in $\cG$ with local coordinate 
$(j,x)$ and the final  and initial vertices of a walk of the form \eqref{walk:def}
\begin{equation}\label{dini}
d(p, v_0(\bw)):= \begin{cases} x & \text{if}\quad p\cong (j,x),\quad 
v_0(\bw)=\partial^-(j),\\ a_j - x & \text{if}\quad p\cong (j,x),\quad 
v_0(\bw)=\partial^+(j),
\end{cases}
\end{equation}
and similarly
\begin{equation}\label{dtermi}
d(q, v_n(\bw)) := \begin{cases} x^\prime & \text{if}\quad q
\cong (j^\prime,x^\prime),\quad  v_n(\bw)=\partial^-(j^\prime),\\ 
a_{j^\prime} - x^\prime & \text{if}\quad q
\cong (j^\prime,x^\prime),\quad v_n(\bw)=\partial^+(j^\prime).
\end{cases}
\end{equation}
The \emph{score} $\underline{n}(\bw)$ of a walk $\bw$ is the set
$\{n_i(\bw)\}_{i\in\cI}$ with $n_i(\bw)\geq 0$ being the number of times
the walk $\bw$ traverses the internal edge $i\in\cI$ such that 
\begin{equation*}
|\bw| = \sum_{i\in\cI} a_i n_i(\bw)
\end{equation*}
holds. Let $\cW_{j,j^\prime}$, $j,j^\prime\in\cE\cup\cI$ be the (infinite if $\cI
\neq \emptyset$) set of all walks $\bw$ on $\cG$ from $j^\prime$ to $j$. 
Obviously we have the
\begin{lemma}\label{lem:distance}
If $p\cong (j,x)$ and $q\cong (j^\prime,x^\prime)$ and if the edges 
$j,j^\prime$ are not tadpoles, then the distance between 
$p$ and $q$ satisfies 
\begin{equation}\label{distance}
d(p,q)\le
\inf_{\bw\in \cW_{j,j^\prime}}(d(p, v_0(\bw)) +|\bw|+
d(q, v_n(\bw)))
\end{equation}
with equality if $ j\neq j^\prime$.
\end{lemma}
Observe that with this notation $d(p, v_0(\bw))\le a_j$ and $d(q, v_n(\bw))\le a_{j^\prime}$.

Using relation (3.33) in \cite{KS8}, relation \eqref{defs} may be rewritten as 
\begin{equation}\label{alternative}
\begin{pmatrix}S(\sk)\\\alpha(\sk)\\
\e^{-\ii\sk\underline{a}}\beta(\sk)\end{pmatrix}=
(\1- \mathfrak{S}(\sk)T(\sk))^{-1}\mathfrak{S}(\sk)\begin{pmatrix} \1_{n\times n} \\
0_{m\times n} \\ 0_{m\times n}
\end{pmatrix}.
\end{equation}
For the sake of clarity we have indicated the type of matrices with $n=|\cE|,m=|\cI|$.
Also $\mathfrak{S}(\sk)=\mathfrak{S}(\sk;A,B)$, see \eqref{mathfracs}, and 
\begin{align*}
T(\sk)=T(\sk,\underline{a})&=\begin{pmatrix}0&0&0\\
                                  0&0&\e^{\ii\sk\underline{a}}\\
               0&\e^{\ii\sk\underline{a}}&0
               \end{pmatrix}.
\end{align*}
So $S(\sk)$ alone is obtained as
\begin{equation}\label{srep}  
S(\sk)=\begin{pmatrix} \1_{n\times n} & 0_{n\times m} & 0_{n\times m} \end{pmatrix}
(\1- \mathfrak{S}(\sk)T(\sk))^{-1}\mathfrak{S}(\sk) 
\begin{pmatrix} \1_{n\times n} \\
0_{m\times n} \\ 0_{m\times n}
\end{pmatrix}.
\end{equation}
Alternative ways of obtaining $S(\sk)$ out of the single vertex scattering matrices $\mathfrak{S}(\sk)$ 
and the metric structure of $\cG$ are given in \cite{Harmer,KSS,KS4, Ragoucy}.
Analogous relations for the amplitudes $\alpha(\sk)$ and $\beta(\sk)$ are 
\begin{align}\label{abrep} 
\alpha(\sk)&=\begin{pmatrix} 
0_{m\times n} & 1_{m\times m} & 0_{m\times m} \end{pmatrix}
(\1- \mathfrak{S}(\sk)T(\sk))^{-1}\mathfrak{S}(\sk) 
\begin{pmatrix} \1_{n\times n} \\
0_{m\times n} \\ 0_{m\times n}
\end{pmatrix}\\\nonumber
\beta(\sk)&=\begin{pmatrix} 
0_{m\times n} & 0_{m\times m} & \e^{\ii\sk\underline{a}} \end{pmatrix}
(\1- \mathfrak{S}(\sk)T(\sk))^{-1}\mathfrak{S}(\sk) 
\begin{pmatrix} \1_{n\times n} \\
0_{m\times n} \\ 0_{m\times n}
\end{pmatrix}.
\end{align}
As a consequence of relation \eqref{srep} the expansion
\begin{equation}\label{walks}
S(\sk)_{ee^\prime}=\sum_{\bw\in \cW_{ee^\prime}}S(\bw;\sk)_{ee^\prime}
\e^{\ii\sk|\bw|}
\end{equation}
with
\begin{equation}\label{swalks}
S(\bw;\sk)_{ee^\prime}=\prod_{l=1}^k S(v_l;\sk)_{i_{l}i_{l-1}}
\end{equation}
is valid.
 $S(v;\sk)$ is the single vertex scattering matrix obtained from the boundary 
conditions at the vertex $v$. Also this matrix is indexed by those edges having 
$v$ in their boundary, that is by the edges in the star graph $\cS(v)$. For this we have to assume 
that there are no {\it tadpoles},
that is edges whose endpoints are the same vertex.
For the details on the expansion \eqref{walks}, see \cite{KS8}. 
But then by the same arguments we also 
obtain similar expansions for the amplitudes $\alpha(\sk)$ and $\beta(\sk)$.
Indeed, for $i\in\cI$ and $e\in\cE$ let $\cW_{ie}^\pm$ be the set of walks in $\cW_{ie}$ 
such that $v_0(\bw)=\partial^\pm(i)$. $\cW_{ie}^-$ and $\cW_{ie}^+$ are disjoint and 
$\cW_{ie}=\cW_{ie}^-\cup\cW_{ie}^+$.
Then \eqref{abrep} implies
\begin{align}\label{absumrep} 
\alpha(\sk)_{ie}&=\sum_{\bw\in \cW_{ie}^-}S(\bw;\sk)_{ie}
\e^{\ii\sk|\bw|}\\\nonumber
\beta(\sk)_{ie}&=\sum_{\bw\in \cW_{ie}^+}S(\bw;\sk)_{ie}
\e^{\ii\sk(a_i+|\bw|)}
\end{align}
with otherwise the same notation as in \eqref{swalks}.

\section{Classical solutions of the Klein-Gordon and the wave equation}\label{sec:clKG}


\subsection{Existence and uniqueness of solutions}\label{subsec:uniq}

Fix boundary conditions $(A,B)$ and introduce the D'Alembert wave operator 
$$
\Box_{A,B}=\frac{\partial^2}{\partial t^2}-\Delta_{A,B}.
$$
For given mass $m>0$, by definition the Klein-Gordon operator is $\Box_{A,B}+m^2$, which we will discuss 
first.
\subsubsection{The Klein-Gordon equation}
Our first discussion for the construction of solutions is close to the familiar one in the 
relativistic case.
Namely, assume $m>0$ to be such that $-\Delta_{A,B}+m^2> 0$. Then actually there is $c>0$ such that 
$-\Delta_{A,B}+m^2>c^2\1$ holds. Indeed, with $\varepsilon_{A,B}=\inf \spec -\Delta_{A,B}\le 0$ the relation 
$\varepsilon_{A,B}+m^2>0$ is valid and so the choice $c=1/2(\varepsilon_{A,B}+m^2)$ does the job.
We introduce the self-adjoint energy operator 
\begin{equation}\label{energyop}
h=h_{A,B,m^2}=\sqrt{-\Delta_{A,B}+m^2}.
\end{equation}
By what has 
just been said
$h>c\1$, so $h$ has a bounded inverse, $0<h^{-1}<c^{-1}\1$.  
For any $f\in L^2(\cG)$ define 
\begin{equation}\label{relscalarpr3}
f^{(\pm)}(p,t)=(\e^{\mp\ii h\,t}f)(p)
\end{equation}
which satisfy 
\begin{equation}\label{relscalarpr31}
\pm\ii\frac{\partial}{\partial t}f^{(\pm)}(p,t)=hf^{(\pm)}(p,t)
\end{equation}
provided $f\in \cD(h)$. 
Moreover both $f^{(\pm)}(p,t)$ satisfy the Klein-Gordon equation
\begin{equation}\label{relscalarpr4}
\left(\Box_{A,B}+m^2\right)f^{(\pm)}(p,t)=0
\end{equation}
provided the stronger initial condition $f\in \cD(-\Delta_{A,B})$ is valid. 
Indeed, since $\cD(-\Delta_{A,B})$ is left 
invariant under $\exp(\mp\ii th)$ ($h$ and $-\Delta_{A,B}$ trivially commute), 
the functions $f^{(\pm)}(p,t)$ are in $\cD(-\Delta_{A,B})$ for all times. 
Due to the choice of the sign in \eqref{relscalarpr3}, $f^{(+)}(p,t)$ is called a positive energy 
solution and $f^{(-)}(p,t)$ a negative energy solution of the Klein-Gordon equation with initial condition 
$f^{(\pm)}(p,t=0)=f(p)$.

For any $g\in L^2(\cG)$, let $g^{(\pm)}(p,t)$ be defined similarly to $f^{(\pm)}(p,t)$. 
If in addition $g\in\cD(-\Delta_{A,B})$, an easy calculation shows that 
\begin{align}\label{relscalarpr5}
\pm \ii(f^{(\pm)}(\cdot,t),\mathop{\partial_t}^{\leftrightarrow}g^{(\pm)}(\cdot,t))_\cG&
=\langle f,g\rangle_\cG\\
(f^{(\pm)}(\cdot,t),\mathop{\partial_t}^{\leftrightarrow}g_\mp(\cdot,t))_\cG&=0.
\end{align} 
holds for all $t$. In the standard context for the Klein-Gordon equation in Minkowski space this result  
is well known, see e.g. \cite{Schweber}, sec. 3b. In particular, the last relation is read as an 
orthogonality relation between positive and negative energy solutions.

We can use these observations to solve the initial problem for the 
hyperbolic differential equation defined by the operator $\Box_{A,B}+m^2$ within the 
$L^2$ context. Indeed, for given $f,\dot{f}$ with 
$f\in \cD(-\Delta_{A,B})=\cD(h^2)$ and $\dot{f}\in \cD(h)$
we will provide  a solution $f(p,t)$ to the Klein-Gordon equation satisfying the initial conditions
\begin{equation}\label{initialcond}
f(p,t=0)=f(p),\qquad \partial_tf(p,t=0)=\dot{f}(p).
\end{equation}
Following standard notation, we call the pair $(f,\dot{f})$ \emph{Cauchy data} for the 
Klein-Gordon equation. In fact with the choice
$$
f^{(\pm)}=\frac{1}{2}\left(f\pm\ii h^{-1}\dot{f}\right) \in \cD(-\Delta_{A,B})
$$
the function 
\begin{equation}\label{KGsol}
f(p,t)=(\e^{-\ii h\,t}f^{(+)})(p)+(\e^{\ii h\,t}f^{(-)})(p)
\end{equation}
solves the initial condition \eqref{initialcond} and satisfies the Klein-Gordon equation.
We make the convention to say that $f(p,t)$ is a solution for all times if for all 
$t\;f(\,,t)\in \cD(-\Delta_{A,B})$
holds,  $f(\,,t)$ is twice differentiable w.r.t. $t$ in the strong topology 
in $L^2(\cG)$ and $\partial_tf(\,,t)\in \cD(h)$ and 
finally if $f(p,t)$ satisfies the Klein-Gordon equation. Similarly we speak of a solution for small 
times if these properties only hold when $|t|<\varepsilon$ for some $\varepsilon>0$.
Obviously $f(p,t)$ as given by \eqref{KGsol} is a solution for all times. 

In standard contexts there is the well known uniqueness of solutions of hyperbolic differential 
equations for given Cauchy data. The standard proof uses energy conservation, see e.g. 
\cite{Evans,Taylor, TaylorI}. In the present context we have 
\begin{proposition}\label{prop:cauchyunique}
For given boundary conditions $(A,B)$ let $m>0$ be such that $-\Delta_{A,B}+m^2>0$.
Set $h=\sqrt{-\Delta_{A,B}+m^2}$ and let Cauchy data $(f,\dot{f})$ be given with 
$f\in \cD(-\Delta_{A,B})$ and $\dot{f}\in \cD(h)$. Then the solution 
for small times exists, is unique, therefore extendable to all times and of the form \eqref{KGsol}. 
\end{proposition}
\begin{proof}
For any solution $g(p,t)$ (for small times) we introduce the energy form 
\begin{equation}\label{energy}
0\le E(g(\,,t))=\langle \partial_tg(\,,t),\partial_tg(\,,t)\rangle_\cG
+\langle hg(\,,t),hg(\,,t)\rangle_\cG.
\end{equation}
Since the scalar product $\langle\,,\,\rangle_\cG$ on $L^2(\cG)$ is positive definite and since $h>c\1>0$, 
for given $t$ $E(g(\,,t))=0$ holds if and only if $g(\,,t)=\partial_tg(\,,t)=0$. 
Also $E(g(\,,t))$ is conserved
\begin{align}\label{energytindep}
\frac{d}{dt}E(g(\,,t))&=\langle \partial_t^2g(\,,t),\partial_tg(\,,t)\rangle_\cG+
\langle \partial_tg(\,,t),\partial_t^2g(\,,t)\rangle_\cG\\\nonumber&\quad +
\langle h\partial_tg(\,,t),hg(\,,t)\rangle_\cG+
\langle hg(\,,t),h\partial_tg(\,,t)\rangle_\cG\\\nonumber
&=-\langle(-\Delta_{A,B}+m^2)g(\,,t),\partial_tg(\,,t)\rangle_\cG 
-\langle \partial_tg(\,,t),(-\Delta_{A,B}+m^2)g(\,,t)\rangle_\cG \\\nonumber
&\quad +\langle \partial_tg(\,,t),(-\Delta_{A,B}+m^2)g(\,,t)\rangle_\cG +
\langle(-\Delta_{A,B}+m^2)g(\,,t), \partial_tg(\,,t)\rangle_\cG \\\nonumber
&=0.
\end{align}
We use this as follows. Let $f_1(p,t)$ and $f_2(p,t)$ be two solutions for small times 
for the same Cauchy data $(f,\dot{f})$ and set $g=f_1-f_2$. By assumption and linearity $g(p,t)$ is also 
a solution for small times.
Moreover $g$ has vanishing Cauchy data, $g(\,,t=0)=\partial_tg(\,,t=0)=0$, which implies 
$E(g(\,,t=0))=0$. But this in turn implies $E(g(\,,t))=0$ for all small $t$ by \eqref{energytindep} 
and therefore $g(\,,t)=\partial_tg(\,,t)=0$ for all small $t$.  
\end{proof}
Concerning the existence of solutions for given initial data the positivity condition
$-\Delta_{A,B}+m^2>0$ may actually be dropped at the price of stronger domain conditions. 
To see this, we use operator calculus in combination with 
the spectral theorem to rewrite the solution \eqref{KGsol} to the Klein-Gordon equation as
\begin{equation}\label{KGsolnew} 
f(\cdot,t)=\cos ht\; f+\frac{\sin ht}{h}\;\dot{f} 
\end{equation}
where both $\cos ht$ and $\sin ht/h$ are bounded self-adjoint operators for all 
real $t$. The solutions at different times $s$ and $t$ are then related by 
\begin{equation}\label{KGsolnew1}
f(\cdot,t)=\frac{\sin h(t-s)}{h}\mathop{\partial_s}^{\leftrightarrow}f(\cdot,s)
\end{equation}
and we observe that this last relation indeed makes sense without the positivity condition 
$-\Delta_{A,B}+m^2> 0$. More precisely, for any boundary condition 
$(A,B)$ and mass $m>0$ introduce the Klein-Gordon kernel
\begin{equation}\label{KGsolnew2}
G_{A,B,m^2}(t)=\frac{\sin\sqrt{-\Delta_{A,B}+m^2}\,t}{\sqrt{-\Delta_{A,B}+m^2}}
\end{equation}
which is well defined by operator calculus. In fact, for fixed $t$ and $m\ge 0$ the functions
\begin{equation}\label{fanal}
z\quad \mapsto\quad \frac{\sin\sqrt{z+m^2}\,t}{\sqrt{z+m^2}},
\quad z\quad \mapsto\quad \cos\sqrt{z+m^2}\,t
\end{equation}
are entire in $z\in\C$ and bounded and real on the real axis. So both $G_{A,B,m^2}(t)$ and 
$\partial_tG_{A,B,m^2}(t)$ are bounded self-adjoint operators for all $t$ and all $m\ge 0$. 
In order to avoid extra superfluous discussion for the case $m=0$ we also make the convention 
\begin{equation}\label{mconv}
\frac{\sin\sqrt{\sk^2+m^2}\,t}{\sqrt{\sk^2+m^2}}\Big|_{m=0}=\frac{\sin \sk t}{\sk}.
\end{equation}
To sum up, 
\begin{equation}\label{KGsolnew3}
f(\cdot,t)=\partial_tG_{A,B,m^2}(t)\;f+G_{A,B,m^2}(t)\;\dot{f} 
\end{equation}
is well defined for all $t$. It satisfies the Klein-Gordon equation and solves the initial problem 
if both $f$ and $\dot{f}$ are in $\cD(-\Delta_{A,B})$. Also \eqref{KGsolnew3}extends to 
\begin{equation}\label{KGsolnew4}
f(\cdot,t)=G_{A,B,m^2}(t-s)\mathop{\partial_s}^{\leftrightarrow}f(s), 
\end{equation}
valid for all $t$ and $s$. It generalizes \eqref{KGsolnew1}. So far, we have not been able to prove 
uniqueness of the solution in this general case, namely when $-\Delta_{A,B}+m^2$ is not necessarily 
a positive operator.


\subsubsection{The wave equation}
We turn to a discussion of the wave operator $\Box_{A,B}$. Consider any boundary condition $(A,B)$.
By the discussion in the previous subsection
\begin{equation}\label{waveeq}
f(\cdot,t)=\partial_tG_{A,B,m^2=0}(t)\;f+G_{A,B,m^2=0}(t)\;\dot{f} 
\end{equation}
is a solution of the wave equation $\Box_{A,B}f(p,t)=0$ for given Cauchy data 
$f,\dot{f}\in\cD(-\Delta_{A,B})$. Concerning uniqueness, there is a result analogous to the one for 
the Klein-Gordon equation, see Proposition \ref{prop:cauchyunique}, given as 
\begin{proposition}\label{prop:wavecauchyunique}
Let the boundary conditions $(A,B)$ be such that $-\Delta_{A,B}$ is non-negative and has no zero eigenvalue.
Then the solution \eqref{waveeq} is the unique solution to the wave equation.
\end{proposition}
In terms of the boundary conditions $(A,B)$ Proposition \ref{prop:spectrum} gives a sufficient 
condition for the absence of negative eigenvalues, that 
is $n_+(AB^\dagger)=0$, while Proposition \eqref{prop:zeroev} provides necessary and sufficient 
conditions for the absence of zero as an eigenvalue.
\begin{proof}
Again we use the energy function \eqref{energytindep}, now with the choice 
$h=\sqrt{-\Delta_{A,B}}\ge 0$. By assumption $hg=0$ implies $g=0$. So
again for given $t$ $E(g(\,,t))=0$ holds if and only if $g(\,,t)=\partial_tg(\,,t)=0$. The proof now 
proceeds as the one for Proposition \ref{prop:cauchyunique}.
\end{proof}

\subsection{Finite propagation speed}\label{susec:finite speed}

In this subsection we will assume the boundary conditions $(A,B)$ to be such that 
$\Sigma^>_{A,B}$ is empty and that zero is not an eigenvalue of $-\Delta_{A,B}$. 
The aim is to analyze support properties of the integral kernel of the operator $G_{A,B,m^2}(t)$. 
When $m>0$ we set $\omega(\sk)=\sqrt{\sk^2+m^2}$.
The completeness relation \eqref{compl} gives
\begin{align}\label{Gkernel} 
G_{A,B,m^2}(t)(p,q)&=\frac{1}{4\pi}\sum_l\int_{-\infty}^\infty d\sk\; \psi^l(p;\sk)\overline{\psi^l(q;\sk)}
\frac{\sin\omega(\sk)\,t}{\omega(\sk)}\\\nonumber
&\qquad+\sum_{\sk\in\Sigma,1\le \nu\le n(\sk)}\psi^{\sk,\nu}(p)
\overline{\psi^{\sk,\nu}(q)}\frac{\sin\omega(\sk)\,t}{\omega(\sk)}.
\end{align}
By our convention \eqref{mconv}, when $m=0$ this simplifies to 
\begin{align}\label{Gkernelm} 
G_{A,B,m^2=0}(t)(p,q)&=\frac{1}{4\pi}\sum_l\int_{-\infty}^\infty d\sk\; \psi^l(p;\sk)\overline{\psi^l(q;\sk)}
\frac{\sin\sk t}{\sk}\\\nonumber
&\qquad+\sum_{\sk\in\Sigma,1\le \nu\le n(\sk)}\psi^{\sk,\nu}(p)
\overline{\psi^{\sk,\nu}(q)}\frac{\sin\sk\,t}{\sk}.
\end{align}
Observe that due to the self-adjointness of $-\Delta_{A,B}$ and as is obvious from \eqref{Gkernel} and 
\eqref{Gkernelm}, the relation 
\begin{equation}\label{sa}
\overline{G_{A,B,m^2}(t)(p,q)}=G_{A,B,m^2}(t)(q,p)
\end{equation}
holds for all $m\ge 0$. In addition, due to \eqref{complexconj} the relation  
\begin{equation}\label{gbar}
\overline{G_{A,B,m^2}(t)(p,q)}=G_{\bar{A},\bar{B},m^2}(t)(q,p)
\end{equation}
is valid. As a consequence, for real boundary conditions $G_{A,B,m^2}(t)(q,p)$ is real.

We define the space of events to be $\R\times \cG$ and write an event as $(t,p)$. 
By definition two events $(t,p)$ and $(s,q)$ are 
\emph{space like separated} if $d(p,q)>|t-s|$.
\begin{theorem}\label{theo:mfinite} 
Assume one of the following two conditions is satisfied.
\begin{itemize}
\item{$\cG$ is a single vertex graph ($\cI=\emptyset$),}
\item{$\cG$ is arbitrary and $-\Delta_{A,B}$ has no discrete eigenvalues.}
\end{itemize}
Then for any $m\ge 0$ the integral kernel $G_{A,B,m^2}(t-s)(p,q)$ vanishes whenever $(t,p)$ and $(s,q)$
are space like separated and if in addition at least one of the two points $p$ and $q$ is in $\cG_{ext}$.
\end{theorem}
So far we have not been able to remove the restriction that $p$ or $q$ must lie in in $\cG_{ext}$.
As a particular case we obtain 
\begin{corollary}\label{pdist}
$G_{A,B,m^2}(t)(p,q)$ vanishes for all $p\in I_e\subset\cG_{ext}$ and all $q\in I_{e^\prime}\subset\cG_{ext}$
whenever $t>0$ is smaller than the passage distance, $t<pdist(e,e^\prime)$.
\end{corollary}
For the free fields to be constructed in the next section this implies local commutativity 
(with the above restriction). We reformulate finite propagation speed in a 
more familiar form. For any closed subset $\cO$ of $\cG$ and any $0<d$ define 
$$
\cO^d=\{p\in\cG\,|\, \min_{q\in\cO}d(p,q)\le d\},
$$
the closed set of points in $\cG$ with distance less or equal to $d$ from $\cO$. 
\begin{corollary}\label{theo:causalsupport}
Under the conditions of the theorem the following holds for the solution 
of the Klein-Gordon equation (or the wave equation) for given Cauchy data $(f,\dot{f})$.
\begin{itemize}
\item{If $f$ and $\dot{f}$ both  
have support in $\cO\subset\cG_{\mathrm{ext}}$, then $f(\cdot,t)$ has support in $\cO^{|t|}$ for all $t$.}
\item{If $f$ and $\dot{f}$ both  
have support in $\cO$, then $\supp f(\cdot,t)\cap \cG_{\mathrm{ext}}\subset \cO^{|t|}$ for all $t$.}
\end{itemize}
\end{corollary}
In particular if both $f$ and $\dot{f}$ have support on the external edge $I_e$, then $f(\cdot,t)$ 
vanishes on any external edge $I_{e^\prime}\,(e^\prime\neq e)$ as long as $|t|<\pdist(e,e^\prime)$.

\section{Free Quantum Fields on Metric Graphs}\label{sec:relnew}
In this section we will construct free fields on the graph $\cG$. The reader is supposed 
to be familiar with the basic concepts of \emph{second quantization}, see, e.g. 
\cite{IZ,Jost,Schweber, Weinberg}.
Also from now on we will assume that the boundary conditions $(A,B)$ 
are chosen in such a way that there are no positive (or zero) eigenvalues of $-\Delta_{A,B}$, 
that is $\Sigma^>=\emptyset$ and $\Sigma=\Sigma^<$, so bound states are still allowed. As a trivial 
consequence of this assumption, the graph has to have 
at least one external edge, $\cE\neq \emptyset$, since otherwise the entire spectrum is discrete and 
there are positive eigenvalues. Finally we will assume that $m>0$ is chosen 
such that $-\Delta_{A,B}+m^2>0$.
\subsection{Creation and annihilation operators and the RT-algebra}
We introduce the creation and annihilation operators \footnote{We stick to the 
standard notational convention in QFT and use $^\star$ to denote the adjoint (only) in this case.}
\begin{align}\label{anncrea}
a^l(\sk),\;a^{l}(\sk)^\star,&\quad \sk>0\\\nonumber
a^{\sk,\nu},\;a^{\sk,\nu\,\star},&\quad \sk\in\Sigma,\;1\le \nu\le n(\sk)
\end{align}
satisfying the commutation relations
\begin{equation}\label{commrel}
\left[a^{l}(\sk),a^{l^\prime}(\sk^\prime)^\star\right]
=2\pi \delta_{ll^\prime}\delta(\sk-\sk^\prime),\qquad
\left[a^{\sk,\nu},a^{\sk^\prime,\nu^\prime\,\star}\right]
=\delta_{\sk,\sk^\prime}\delta_{\nu,\nu^\prime}
\end{equation}
while all other commutators vanish. These operators act in the bosonic Fock space $\fF(\cH_1)$ 
with $\cH_1=L^2(\cG)$ as the choice of the 1-particle space, that is 
\begin{align}
\fF(\cH_1)&=\C\oplus \cH_1\oplus\cdots \oplus\cH_n\oplus\cdots\\
       \cH_n&=\underbrace{\cH_1\mathop{\otimes}_{s}\cH_1\mathop{\otimes}_{s}\cH_1
\mathop{\otimes}_{s}\cH_1}_n,
\end{align}
such that $\cH_n$ is the $n$-particle space. $\mathop{\otimes}_{s}$ denotes the symmetric 
tensor product. $a^{l}(\sk)^\star$ has the interpretation of a creation of a particle with 
wave function $\psi^l(\,;\sk)$, while $a^{\sk,\nu\,\star}$ is the creation operator of a particle with 
(bound state) wave function $\psi^{\sk,\nu}$. The normalization in \eqref{commrel} is chosen in 
accordance with \eqref{ortho}, \eqref{posevform} and \eqref{efnormal}.
For reasons which will become clear in a moment, 
we elaborate on this. 
By the completeness relation \eqref{compl} any wave 
function $f\in L^2(\cG)$ has a \emph{Fourier type} expansion of the form
\begin{equation}\label{fexpans}
f(p)=\sum_{l\in\cE}\int_{0}^\infty d\sk \widetilde{f}_l(\sk)\psi^l(p;\sk)
+\sum_{\sk\in\Sigma, 1\le \nu\le n(\sk)}\widetilde{f}^\nu(\sk)\psi^{\sk,\nu}(p)
\end{equation}
with expansion coefficients given as 
\begin{equation}\label{fexpans1}
\widetilde{f}_l(\sk)=\frac{1}{\sqrt{2\pi}}\int_{p\in \cG}\overline{\psi^l(p;\sk)}f(p)dp,\qquad 
\widetilde{f}^\nu(\sk)=\int_{p\in \cG}\overline{\psi^{\sk,\nu}(p)}f(p)dp
\end{equation}
such that the \emph{Parseval equality} holds in the form
\begin{equation}\label{fexpans2}
\langle f,f\rangle_\cG=\sum_{l\in\cE}\int_{0}^\infty d\sk|\widetilde{f}_l(\sk)|^2\;+
\sum_{\sk\in\Sigma, 1\le \nu\le n(\sk)}|\widetilde{f}^\nu(\sk)|^2
\end{equation}
holds thus establishing an isometry of Hilbert spaces
$$
L^2(\cG)\cong L^2([0,\infty),d\sk)\,\oplus\, \C^{N_\Sigma}
$$
where 
$$
N_\Sigma=\sum_{\sk\in\Sigma=\Sigma^<}n(\sk)\le |\cE|+2|\cI|
$$
is the total number of bound states, counting multiplicities.
With this notation the creation operator for a particle with an arbitrary wave function $f$ is of the form
\begin{equation}\label{fexpans3}
a^\star(f)=\sum_{l\in\cE}\int_{0}^\infty d\sk \widetilde{f}_l(\sk)a^l(\sk)^\star
+\sum_{\sk\in\Sigma, 1\le \nu\le n(\sk)}\widetilde{f}^\nu(\sk)a^{\sk,\nu\,\star}
\end{equation}
and correspondingly its adjoint $a(f)$ is the annihilation operator for the wave function $f$.

The (self-adjoint) number operator, the second quantization of the identity operator on the one-particle 
space, is 
\begin{equation}\label{secnumber}
{\bf N}=\frac{1}{2\pi}\sum_{l\in \cE}\int_0^\infty d\sk\; 
a^{l}(\sk)^\star a^{l}(\sk)+\sum_{\sk\in\Sigma, 1\le \nu\le n(\sk)}
a^{\sk,\nu\;\star} a^{\sk,\nu}.
\end{equation}
We define $h=\sqrt{-\Delta_{A,B}+m^2}$, see the discussion in Section  \ref{subsec:uniq}, 
to be the one-particle Hamilton operator, so its second quantization is the self-adjoint operator 
\begin{equation}\label{secham}
{\bf H}=\frac{1}{2\pi}\sum_{l\in \cE}\int_0^\infty d\sk\; 
\omega(\sk) \,a^{l}(\sk)^\star a^{l}(\sk)+\sum_{\sk\in\Sigma, 1\le \nu\le n(\sk)}\omega(\sk)
a^{\sk,\nu\;\star} a^{\sk,\nu}
\end{equation}
and we observe that $\omega(\sk)$ is positive for all $\sk\in\Sigma^<=\Sigma$ by the choice of $m>0$.

The operator
\begin{equation}\label{momrel}
{\bf P}=\frac{1}{2\pi}\sum_{l\in \cE}\int_0^\infty d\sk\; 
\sk \,a^{l}(\sk)^\star a^{l}(\sk)
\end{equation}
can be given the interpretation of the sum of the absolute value of the momenta of all particles in a 
state of the Fock space which does not contain particles with bound state wave functions.
Stated more abstractly, let $P_{ac}$ be the orthogonal projector onto the subspace of $L^2(\cG)$ 
corresponding to the absolutely continuous spectrum of $-\Delta_{A,B}$. 
Then $\bf P$ is the second quantization of the 1-particle operator 
$\sqrt{-\Delta_{A,B}P_{ac}}$. That there is no proper momentum operator in the familiar sense 
has of course to do with the fact that the configuration space is a graph. So the notion of translations
in space and with the momentum operator as infinitesimal generator does not make sense. 
But what remains is some kind of absolute value of momentum reminiscent of the conservation of the 
absolute value of the momentum of a (classical) particle under elastic scattering.
Both ${\bf N}$ and ${\bf P}$ commute with $H_0$ and are therefore conserved under time evolution. 

With these preparatory remarks we are now in the position to provide an explicit construction of 
\emph{RT (reflection-transmission)-algebras} 
\cite{CMRS,MRSI,MRSII}. The main observation is that $\sk$ in $a^{l}(\sk)$ and $a^{l}(\sk)^\star$ is 
positive. So we are \emph{free to define} creation and annihilation 
operators also for negative $\sk$. Indeed, we may set
\begin{align}\label{extension1}
a^{l}(-\sk)&=\sum_{l^\prime\in\cE}S(\sk)_{l\,l^\prime}a^{l^\prime}(\sk)\\\nonumber
a^{l}(-\sk)^\star&=\sum_{l^\prime\in\cE}S(-\sk)_{l^\prime \,l}a^{l^\prime}(\sk)^\star,
\qquad \sk>0,
\end{align}
where we recall the general relation $S(-\sk)=S(\sk)^{-1}=S(\sk)^\dagger$ valid for all real $\sk\neq 0$. 
With this definition the relations \eqref{extension1}
remain valid for $\sk<0$ and 
then $a^{l}(\sk)^\star$ is again the adjoint of $a^{l}(\sk)$. Since for $\sk>0$ the operator 
$a^{l^\prime}(\sk)^\star$  creates a particle with wave function $\psi^l(\,;\sk)$, 
by linearity the operator $a^{l}(-\sk)^\star$ as defined by \eqref{extension1} 
creates a particle with wave function
\begin{equation*} 
\sum_{l^\prime\in\cE}S(-\sk)_{l^\prime\,l}\psi^{l^\prime}(p;\sk), 
\end{equation*}
which by \eqref{alphabeta4} equals $\psi^l(\,;-\sk)$. This gives the first part of  the next lemma, 
while the second part follows by an easy calculation.
\begin{lemma}\label{a-k}
For any $\sk>0$ the operator $a^{l}(-\sk)^\star$ as defined by \eqref{extension1},  
creates a particle with wave function $\psi^l(\,;-\sk)$.
The extended family of operators 
$$\left\{a^l(\sk),a^l(\sk)^\star\right\}_{l\in\cE,-\infty<\sk<\infty}
$$
satisfies the commutation relations
\begin{equation}\label{extension2}
\left[a^{l}(\sk),a^{l^\prime}(\sk^\prime)^\star\right]=\delta_{l\,l^\prime}\delta(\sk-\sk^\prime)+
S(\sk)_{ll^\prime}\delta(\sk+\sk^\prime),\qquad -\infty<\sk,\sk^\prime <\infty,\quad l,l^\prime\in\cE,
\end{equation}
again with all other commutators vanishing.
\end{lemma}
\begin{remark}\label{re:rtalgebra}
This realization of a RT-algebra agrees with the one used in  
\cite{BellBurrMinSo,Ragoucy}. The construction \eqref{extension1} of the $a^l(-\sk)$ and $a^l(-\sk)^\star$ 
out of the $a^{l^\prime}(\sk)$ and $a^{l^\prime}(\sk)^\star$ is 
reminiscent of the action of the Weyl group in the root space of a Lie algebra, by which any root is 
obtained  from the set of positive roots \cite{Faddeev}. A different context, where a (scalar) scattering 
matrix appears in commutation relations, is provided in \cite{Faddeev:Volkov}. 
\end{remark}

\subsection{The free  hermitian quantum field}\label{subsec:hermfield}
For reasons to become clear in a moment, in this subsection the boundary conditions $(A,B)$ 
will be taken to be real. 
The field operator, again of dimension zero, is defined to be 
\begin{align}\label{phipt}
\Phi(t,p)&=\e^{\ii {\bf H}\,t}\Phi(p)\e^{-\ii  {\bf H}\,t}\\\nonumber
&=\sum_{l\in\cE}
\int_0^\infty 
\frac{d\sk}{\sqrt{2\pi}}\frac{1}{\sqrt{2\omega(\sk)}}
\left(\overline{\psi^l(p;\sk)}\,\e^{\ii  \omega(\sk)t}a^{l}(\sk)^\star
+h.c.\right)\\\nonumber
&\qquad\qquad +\sum_{\sk\in\Sigma,1\le \nu\le n(\sk)}
\frac{1}{\sqrt{2\omega(\sk)}}
\left(\overline{\psi^{\sk,\nu}(p)}\,\e^{\ii \omega(\sk)t}a^{\sk,\nu\,\star}+h.c.\right).
\end{align}
where $h.c.$ denotes hermitian conjugate. By construction, this field is hermitian and 
$\Phi(t+s,p)=\e^{\ii {\bf H}\,t}\Phi(s,p)\e^{-\ii  {\bf H}\,t}$ holds.
Again we use a similar notational convention as the one 
used for a local description of functions on $\cG$. 
Thus for its restriction to an edge $j$ and with local coordinate $(j,x)\,x\in [0,a_j]$ for a point $p$ 
there the field is given as 
\begin{equation}\label{phiptlocal}
\Phi_j(t,x)=\begin{cases}\sum_{l\in\cE}
\int_0^\infty\frac{d\sk}{\sqrt{2\pi}}\frac{1}{\sqrt{2\omega(\sk)}}
\Big(\left(\e^{\ii \sk x}\delta_{jl}
+\overline{S(\sk)_{jl}}\e^{-\ii \sk x}\right)
\e^{\ii \omega(\sk)t}\,a^{l}(\sk)^\star+h.c.\Big)\\
\\
+\sum_{\sk\in\Sigma,1\le \nu\le n(\sk)}
\frac{1}{\sqrt{2\omega(\sk)}}
\left(\overline{s^{\sk,\nu}_j}\e^{\ii\sk x}
\e^{\ii\omega(\sk)t}\,a^{\sk,\nu\,\star}+h.c.\right)
,\qquad j\in\cE\\
\\
\sum_{l\in\cE}\int_0^\infty \frac{d\sk}{\sqrt{2\pi}}\frac{1}{\sqrt{2\omega(\sk)}}
\Big(\left(\overline{\alpha(\sk)_{jl}}\e^{-\ii\sk x}
+\overline{\beta(\sk)_{jl}}\e^{\ii\sk x}  \right)
\e^{\ii \omega(\sk)t}\,a^{l}(\sk)^\star +h.c.\Big)\\
\\
+\sum_{\sk\in\Sigma,1\le \nu\le n(\sk)}
\frac{1}{\sqrt{2\omega(\sk)}}\left(\left( \overline{u^{\sk,\nu}_j}\e^{\ii\sk x}+
\overline{v^{\sk,\nu}_j}\e^{-\ii\sk x}\right)\e^{\ii\omega(\sk)t}\,a^{\sk,\nu\,\star}+h.c.\right),
\; j\in\cI.
\end{cases}
\end{equation}
Observe that the $\psi^{\sk,\nu}(p)$ need not 
be chosen real. However, the reality of the boundary conditions comes as follows into play. 
By Corollary \ref{cor:laplacereal2} the $\overline{\psi^l(p;\sk)}$ and the 
$\overline{\psi^{\sk,\nu}(p)}$ are also eigenfunctions of $-\Delta_{A,B}$.
Since the boundary conditions are real, we can use Lemma \ref{lem:reales} and \eqref{extension1} 
to simplify the first terms in \eqref{phiptlocal} using the RT-algebra notation and 
$\omega(-\sk)=\omega(\sk)$
\begin{align}\label{phiptlocalrt}
\sum_{l\in\cE}
\int_0^\infty\frac{d\sk}{\sqrt{2\pi}}\frac{1}{\sqrt{2\omega(\sk)}}&
\Big(\left(\e^{\ii\sk x}\delta_{jl}
+\overline{S(\sk)_{jl}}\e^{-\ii \sk x}\right)
\e^{\ii \omega(\sk)t}\,a^{l}(\sk)^\star+h.c.\Big)=\\\nonumber
&
\int_{-\infty}^\infty\frac{d\sk}{\sqrt{2\pi}}\frac{1}{\sqrt{2\omega(\sk)}}
\Big(\e^{\ii (\sk x+\omega(\sk)t)}\,a^{j}(\sk)^\star+h.c.\Big)
,\qquad\qquad\quad\; j\in\cE\\\nonumber
\sum_{l\in\cE}\int_0^\infty \frac{d\sk}{\sqrt{2\pi}}\frac{1}{\sqrt{2\omega(\sk)}}&
\Big(\left(\overline{\alpha(\sk)_{jl}}\e^{-\ii\sk x}
+\overline{\beta(\sk)_{jl}}\e^{\ii\sk x}  \right)
\e^{\ii \omega(\sk)t}\,a^{l}(\sk)^\star+h.c.\Big)=\\\nonumber
&\sum_{l\in\cE}
\int_{-\infty}^\infty\frac{d\sk}{\sqrt{2\pi}}\frac{1}{\sqrt{2\omega(\sk)}}
\Big(\beta(-\sk)_{jl}\e^{\ii(\sk x+\omega(\sk)t)} 
\,a^{l}(\sk)^\star+h.c.\Big),\quad j\in\cI.
\end{align}
Let $\Omega$ denote the vacuum.
\begin{proposition}\label{prop:bdyphi}
The hermitian field $\Phi$ satisfies the Klein- Gordon equation
$$
(\Box_{A,B}+m^2)\Phi(p,t)=0.
$$
For all times $t$ the boundary conditions
$$
A\underline{\Phi}(t)+B\underline{\Phi}^\prime(t)=0
$$
 are valid in the sense of expectation values in states which are 
linear combinations of states of the form 
$$
\prod_i a(f_i)^\star\Omega
$$
with $f_i\in\cD(-\Delta_{A,B}+m^2)$.
\end{proposition}
\begin{proof}
For general boundary conditions $(A,B)$ we recall that if $\psi$ satisfies the 
boundary 
condition \eqref{bdycond}, then $\bar{\psi}$ satisfies the boundary condition 
$(\bar{A},\bar{B})$. As a consequence, if the boundary conditions $(A,B)$ are real, 
then both $\e^{\ii\omega(\sk)t}\psi^l(\,;\sk)$ and 
$\e^{\ii\omega(\sk)t}\overline{\psi^l(\,;\sk)}$ satisfy the boundary 
condition \eqref{bdycond} for all $l\in\cE$ and all $\sk\in\R_+$ and the claim follows
from the construction of $\Phi$ and the choice of the states. We omit details.
\end{proof}
We also introduce the conjugate field
\begin{equation}\label{pifield}
\Pi(p,t)=\dot{\Phi}(p,t)=\frac{\partial}{\partial t}\Phi(p,t).
\end{equation}
Using the completeness relation for the eigenfunctions of 
$-\Delta_{A,B}$  
in the form \eqref{compl} we derive the 
\begin{theorem}\label{canonical}
For the boundary conditions $(A,B)$ 
the equal time commutation relation
\begin{equation}\label{equaltime}
\left[\Phi(p,t),\Pi(q,t)\right]=i\delta(p,q), \quad p,q\in\cG
\end{equation}
is valid.
\end{theorem}
Observe that this relation fixes the normalization of the field.

\subsection{The free complex quantum field}\label{subsec:complex}

We now construct a complex field $\Psi$, which has the advantage of being able to 
carry (electric) charge. Associated is a particle with that charge and an antiparticle with the opposite 
charge. Accordingly the 1-particle space $\cH_1$ is chosen to be 
$L^2(\cG)\oplus L^2(\cG)$, 
the first for a particle and the second for the corresponding antiparticle. 

The 1-particle Hamiltonian $h$ on that space is chosen to be 
\begin{equation}\label{relcompl}
h=\sqrt{-\Delta_{A,B}+m^2}\,\oplus\, 0\;
+\;0\,\oplus\, \sqrt{-\Delta_{\bar{A},\bar{B}}+m^2}.
\end{equation}
The boundary conditions $(A,B)$ themselves may be chosen arbitrary.
To simplify the exposition we assume that $-\Delta_{A,B}$ and 
hence also $-\Delta_{\bar{A},\bar{B}}$ has no discrete spectrum, cf. Corollary 
\ref{conjspectrum}. So in particular $-\Delta_{A,B}\ge 0,\;-\Delta_{\bar{A},\bar{B}}\ge 0$.
Since the boundary conditions $(A,B)$ are not necessarily real, relation \eqref{realimprop} need not hold. 
However, $\overline{\psi(\,;\sk)}$ satisfies the boundary conditions 
$(\bar{A},\bar{B})$ by Lemma \ref{lem:bdycomplex}.
The creation and annihilation operators for the particles are as before, see \eqref{commrel}. 
As for the antiparticles, for $l\in\cE$ and $\sk>0$ introduce operators $b^l(\sk)$ and their adjoints 
$b^l(\sk)^\star$ satisfying commutation relations of the same form and commuting with all 
$a^{l^\prime}(\sk^\prime)$ and $a^{l^\prime}(\sk^\prime)^\star$.
They are the annihilation and creation operators for the antiparticle with wave function 
$\overline{\psi^l(\,;\sk)}=\overline{\psi^l_{A,B}(\,;\sk)}$, which we recall differs from 
$\psi^l_{\bar{A},\bar{B},\underline{a}}(\,;\sk)$. Correspondingly we set 
\begin{align*}
b^{l}(-\sk)&=\sum_{l^\prime\in\cE}S(-\sk)_{l^\prime\,l}b^{l^\prime}(\sk)
=\sum_{l^\prime\in\cE}\overline{S(\sk)}_{l\,l^\prime}b^{l^\prime}(\sk)\\\nonumber
b^{l}(-\sk)^\star&=\sum_{l^\prime\in\cE}\overline{S(-\sk)}_{l^\prime \,l}b^{l^\prime}(\sk)^\star
               =\sum_{l^\prime\in\cE}S(\sk)_{l\,l^\prime }b^{l^\prime}(\sk)^\star
\qquad \sk>0
\end{align*}
with $S(\sk)=S_{A,B}(\sk)$. By \eqref{alphabeta4} the interpretation 
is that $b^{l}(-\sk)^\star$ creates a particle with wave function 
$\overline{\psi^l(\,;-\sk)}$. ${\bf H}$, the second quantization of $h$ as given by \eqref{relcompl}, is 
\begin{equation*}
{\bf H}= \frac{1}{2\pi}\sum_{l\in \cE}\int_0^\infty d\sk\; 
\omega(\sk) \,a^{l}(\sk)^\star a^{l}(\sk)+\frac{1}{2\pi}\sum_{l\in \cE}\int_0^\infty d\sk\; 
\omega(\sk) \,b^{l}(\sk)^\star b^{l}(\sk).
\end{equation*}
The field $\Psi$ and its adjoint is now given as 
\begin{align}\label{complexrel}
\Psi(t,p)&=\e^{\ii {\bf H}\,t}\Psi(p)\e^{-\ii {\bf H}\,t}
=\sum_{l\in\cE}\int_0^\infty \frac{d\sk}{\sqrt{2\pi}}\frac{1}{\sqrt{2\omega(\sk)}}\,\psi^l(p;\sk)
\left(\e^{\ii  \omega(\sk)t}\,b^{l}(\sk)^\star
+\e^{-\ii \omega(\sk)t}\,a^l(\sk)\right)\\\nonumber
\Psi^\dagger(t,p)&=\e^{\ii {\bf H}\,t}\Psi^\dagger(p)\e^{-\ii {\bf H}\,t}=\sum_{l\in\cE}\int_0^\infty 
\frac{d\sk}{\sqrt{2\pi}}\frac{1}{\sqrt{2\omega(\sk)}}\,
\overline{\psi^l(p;\sk)}
\left(\e^{-\ii \omega(\sk)t}\,b^{l}(\sk)+\e^{\ii \omega(\sk)t}\,a^l(\sk)^\star\right).
\end{align}
In local coordinates and in terms of the RT-algebra we can write the field $\Psi$ (and similarly 
its adjoint)
as  
\begin{equation}\label{localcomplexrel}
\Psi_j(t,x)=\begin{cases}\int_{-\infty}^\infty \frac{d\sk}{\sqrt{2\pi}}\frac{1}{\sqrt{2\omega(\sk)}}\,
\left(\e^{\ii(\omega(\sk)t-\sk x)}\,b^{j}(\sk)^\star
+\e^{-\ii(\omega(\sk)t+\sk x)}a^j(\sk)\right),\quad j\in\cE\\
\\
\sum_{l\in\cE}\int_{-\infty}^\infty \frac{d\sk}{\sqrt{2\pi}}\frac{1}{\sqrt{2\omega(\sk)}}\,\beta_{jl}(\sk)
\left(\e^{\ii(\omega(\sk)t-\sk x)}\,b^{l}(\sk)^\star 
+\e^{-\ii(\omega(\sk)t+\sk x)}a^l(\sk)\right),\quad j\in\cI.
\end{cases}
\end{equation}
Use has been made of \eqref{-k}.
The motivation for this definition of the one particle Hilbert space for the antiparticle, the 
corresponding 1-particle Hamiltonian and finally the field $\Psi$ stems from  
\begin{proposition}\label{prop:bdypsi}
The field $\Psi(p,t)$ and its adjoint $\Psi^\dagger(p,t)$ satisfy the Klein-Gordon equation   
$$
(\Box+m^2)\Psi(p,t)=0,\qquad (\Box+m^2)\Psi^\dagger(p,t)=0
$$
and the boundary conditions 
$$
A\underline{\Psi}(t)+B\underline{\Psi}^{\prime}(t)=0=\bar{A}\underline{\Psi}^\dagger(t)
+\bar{B}\underline{\Psi}^{\dagger\,\prime}(t) 
$$
for all times.
\end{proposition}
As in Proposition \ref{prop:bdyphi} the last relation holds in the sense of an expectation value
in suitable states.

Let  $\cC$ denote charge conjugation, the operation which interchanges particles and antiparticles.
In addition introduce the antilinear and antiunitary time reversal  map $\cT$ , 
cf. Remark \ref{re:timereversal}. 
Then there is $\cC\cT$ invariance, that is 
\begin{equation}\label{ctinv}
\cC\cT\Psi(p,t)(\cC\cT)^{-1}=\Psi(p,-t)
\end{equation}
holds.

\subsection{The commutator function}

In this subsection we calculate the commutator of the fields.
For the hermitian field we obtain 
\begin{align}\label{hermcomm}
\left[\Phi(t,p),\Phi(s,q)\right]&=\sum_{l\in\cE}
\int_0^\infty\frac{d\sk}{2\pi}\frac{1}{2\omega(\sk)}
\left(\overline{\psi^l(p;\sk)}\psi^l(q;\sk)\,\e^{\ii\omega(\sk)(t-s)}-
\psi^l(p;\sk)\overline{\psi^l(q;\sk)}\e^{-\ii\omega(\sk)(t-s)}\right)\\\nonumber
&\quad+\sum_{\sk\in\Sigma,1\le \nu\le n(\sk)}\frac{1}{2\omega(\sk)}
\left(\overline{\psi^{\sk,\nu}(p)}\psi^{\sk,\nu}(q)\e^{\ii\omega(\sk)(t-s)}-
\psi^{\sk,\nu}(p)\overline{\psi^{\sk,\nu}(q)}\e^{-\ii\omega(\sk)(t-s)}\right)
\end{align}
Since the boundary conditions are real, the reality properties  
\begin{equation}\label{eigenproj}
\sum_{l\in\cE}\overline{\psi^l(p;\sk)}\psi^l(q;\sk)=\sum_{l\in\cE}\psi^l(p;\sk)\overline{\psi^l(q;\sk)}
;\; \sum_{1\le \nu\le n(\sk)}\psi^{\sk,\nu}(p)
\overline{\psi^{\sk,\nu}(q)}=\sum_{1\le \nu\le n(\sk)}\overline{\psi^{\sk,\nu}(p)}\psi^{\sk,\nu}(q)
\end{equation}
hold. Indeed, the first relation is easily derived from \eqref{realimprop}. To prove the second one,
observe that for given $\sk\in\Sigma=\Sigma^<$ both sides give the unique integral kernel for the 
orthogonal projector in $L^2(\cG)$ onto the eigenspace of $ -\Delta_{A,B}$ with eigenvalue $\sk^2<0$.
In fact, since the $\psi^{\sk,\nu}$ form an orthonormal basis in that space, so do their complex conjugates.
Inserting the relations \eqref{eigenproj} into \eqref{hermcomm} gives the first part of the next theorem. 
The proof of the second part is even easier and will therefore be omitted.
\begin{theorem}\label{theo:G=Delta}
The commutator for the free hermitian field with real boundary conditions $(A,B)$ is given as 
\begin{equation}\label{hermcomm1}
\left[\Phi(t,p),\Phi(s,q)\right]=-\ii G_{A,B,m^2}(t)(p,q).
 \end{equation}
Similarly for the complex field and arbitrary boundary conditions $(A,B)$ the commutators are 
\begin{equation}\label{hermcomm2}
\left[\Psi(t,p),\Psi(s,q)^\dagger\right]=-\ii G_{A,B,m^2}(t)(p,q),\qquad \left[\Psi(t,p),\Psi(s,q)\right]=0.
\end{equation}
\end{theorem}
The last relation of course also implies $\left[\Psi(t,p)^\dagger,\Psi(s,q)^\dagger\right]=0$.
In the Minkowski space context it is well known that 
(up to a sign) the Klein-Gordon kernel equals the commutator function, see e.g. \cite{Schweber} sec. 7c and 
\eqref{commmink} below. 
So in analogy to the Minkowski space context and as a 
consequence of finite propagation speed we have local commutativity in the form 
\begin{corollary}\label{corr:Einstein}
For space like separated events $(t,p)$ and $(s,t)$ the commutators \eqref{hermcomm1} and \eqref{hermcomm2} 
vanish provided as least one of the points $p$ and $q$ lies in $\cG_{ext}$.
\end{corollary}
\subsection{Examples}We illustrate our discussion in the context of single vertex graphs with two 
simple examples.
First we make the following notational convention. If $p$ has local coordinate $(i,x)$ and $q$ the local 
coordinate $(j,y)$ and for given $(A,B)$ and $m$ we set 
\begin{equation}\label{convention}
G_{ij}(t,x;s,y)=G_{A,B,m^2}(t-s)(p,q).
\end{equation} 
Also we will need the following quantities.
Let $\Delta(t,x;m)$ be the usual relativistic commutator function of mass $m$ 
in $1+1$ space-time dimensions
\begin{align}\label{commmink}
\Delta(t,x;m)&=\frac{1}{2\pi\ii}\int_{-\infty}^\infty\frac{d\sk}{2\omega(\sk)}\e^{\ii\sk x}
\left(\e^{-\ii\omega(\sk)t}-\e^{\ii\omega(\sk)t}\right)
=\Delta^{(+)}(t,x;m)+\Delta^{(-)}(t,x;m)\\
&=-\int_{-\infty}^\infty\frac{d\sk}{2\pi}\e^{\ii\sk x}\frac{\sin\omega(\sk)t}{\omega(\sk)}.
\end{align}
More explicitly
\begin{equation}\label{jordanfunction}
\Delta(t,x;m)=\begin{cases}\;\;0\quad &t^2-x^2<0\\
-\sign\, t\, N_0(m\sqrt{t^2-x^2})\quad &t^2-x^2>0.
\end{cases}
\end{equation}
$N_0$ is the zero'th Neumann function (a Bessel function of the second kind). For large argument it 
satisfies 
\begin{equation}\label{neumann}
N_0(z)\simeq \sqrt{\frac{2}{\pi z}}\sin(z-\pi/4) \qquad\mbox{for}\quad 1\ll z.
\end{equation}
For a more detailed discussion of the commutator function in local coordinates and which will be needed 
in the proof of Theorem \ref{theo:mfinite} in Appendix \ref{app:KGsupport}, introduce the 
distribution in $0<x,-\infty<t<\infty$
\begin{align}\label{sprop1}
\fD(t,x;m,\kappa)&=\frac{1}{2\pi\ii}\int_{-\infty}^\infty\frac{d\sk}{2\omega(\sk)}
\e^{\ii\sk x}\left(\e^{-\ii \omega(\sk)t}-\e^{\ii \omega(\sk)t}\right)\frac{\sk+\ii\kappa}{\sk-\ii\kappa}\\
\nonumber
&=\fD^{(+)}(t,x;m,\kappa)+\fD^{(-)}(t,x;m,\kappa)\\\nonumber
&=-\int_{-\infty}^\infty\frac{d\sk}{2\pi}\e^{\ii\sk x}\frac{\sin\omega(\sk)t}{\omega(\sk)}
\;\frac{\sk+\ii\kappa}{\sk-\ii\kappa}
\end{align}
with $m>0$ and $\kappa$ real, the values $\kappa=0,\infty$ being allowed, that is 
\begin{equation}\label{sprop2}
\fD(t,x;m,\kappa=0)=\Delta(t,x;m),\qquad \fD(t,x;m,\kappa=\infty)=-\Delta(t,x;m).
\end{equation}
By construction $\fD(t,x;m,\kappa)$ is odd in $t$. For $\kappa\neq \infty$ write  
\begin{align}\label{sprop3}
\fD(t,x;m,\kappa)&=\Delta(t,x;m)+\fd(t,x;m,\kappa)\\\nonumber
&=\fD^{(+)}(t,x;m,\kappa)+\fD^{(-)}(t,x;m,\kappa)\\\nonumber
\fD^{(\pm)}(t,x;m,\kappa)&=\Delta^{(\pm)}(t,x;m)+\fd^{(\pm)}(t,x;m,\kappa)
\end{align}
with the \emph{bona fide} function
\begin{align}\label{sprop4}
\fd(t,x;m,\kappa)&=\frac{2\ii\kappa}{2\pi\ii}\int_{-\infty}^\infty\frac{d\sk}{2\omega(\sk)}
\e^{\ii\sk x}\left(\e^{-\ii \omega(\sk)t}-\e^{\ii \omega(\sk)t}\right)\frac{1}{\sk-\ii\kappa}\\\nonumber
&=\fd^{(+)}(t,x;m,\kappa)+\fd^{(-)}(t,x;m,\kappa)\\\nonumber
&=-2\ii\kappa\int_{-\infty}^\infty\frac{d\sk}{2\pi}\e^{\ii\sk x}\frac{\sin\omega(\sk)t}{\omega(\sk)}
\frac{1}{\sk-\ii\kappa}
\end{align}
such that $\fd^{(-)}(t,x;m,\kappa)=-\fd^{(+)}(-t,x;m,\kappa)=\overline{\fd^{(+)}(t,x;m,\kappa)}$. 
It is easy to show that for given $m$ and $\kappa$ 
$\fd^{(\pm)}$ are uniformly bounded functions of $x$ and $t$ and H\"older continuous in both $x$ and $t$ 
of H\"older index $<1$.
$\fD^{(+)},\fd^{(+)}$ and $\fD^{(-)},\fd^{(-)}$ are positive and negative energy solutions of the 
usual Klein-Gordon equation respectively
$$
(\partial_t^2-\partial_x^2+m^2)\fD^{(\pm)}(x,t;m,\kappa)=
(\partial_t^2-\partial_x^2+m^2)\fd^{(\pm)}(x,t;m,\kappa)=0.
$$
Moreover, the differential equation 
\begin{equation}\label{ddiffeq}
\left(\frac{\partial}{\partial x}+\kappa\right)\fd^{(\pm)}(x,t;m,\kappa)=-2\kappa\Delta^{(\pm)}(x,t:m) 
\end{equation}
holds. $\fd^{(\pm)}(t,x;m,\kappa)$ decays at least like 
$|t|^{-1/2}$ for large $t$ and fixed $x$. This is well known from the theory of Haag-Ruelle scattering theory,
see e.g. \cite{ Glimm-Jaffe,Jost}. Sufficient conditions for stronger decay are also well known 
but do not apply here. 
When $0<x<|t|$ the \emph{stationary phase approximation} gives 
\begin{equation}\label{statphase}
\fd^{(\pm)}(t,x;m,\kappa)\cong \left(1-v^2\right)^{-1/4}
\frac{\kappa}{\pm\frac{mv}{\sqrt{1-v^2}}-\ii\kappa}\e^{\mp(\phi(x,t)+\ii\frac{\pi}{4}\sign t)}
\frac{1}{\sqrt{2\pi m|t|}}
\end{equation}
with $v=x/t$ and $\phi(x,t)=m\sqrt{1-v^2}\,t$.
As a function of $t$ and for fixed $x$ the $t^{-1/2}$ decay as well 
as the oscillations are visible in numerical computations of $\fd^{(\pm)}$. 
\begin{example}(The half-line with Robin boundary conditions at the origin)\label{ex:robinex} 
View the positive real axis $\R_+$ as a single vertex graph with one external edge, $|\cE|=1$. 
All possible boundary conditions at the origin giving rise to self-adjoint Copulations are the  
{\rm Robin} boundary conditions and which are real 
\begin{equation}\label{robin}
\cos\, \tau\; \psi(0)+\sin\,\tau\; \psi^\prime(0)=0,\quad 0\le \tau<\pi.
\end{equation}
They interpolate between {\rm Dirichlet} ($\sin\tau=0$) and {\rm Normans} ($\cos\tau=0$) boundary 
conditions. 
\end{example}
Denote the resulting Laplace operator by $-\Delta_\tau$.
The scattering matrix is now just a function 
\begin{equation}\label{stau}
S_\tau(\sk)=-\frac{\cos\,\tau -\ii \sk \sin\,\tau}{\cos\, \tau +\ii \sk \sin\,\tau}
\end{equation}
satisfying $S_\tau(-\sk)=S_\tau(\sk)^{-1}$ for $\sk\in\C$ and being of modulus $1$ for 
$\sk\in\R\setminus\{0\}$, as it should.
There is a pole of $S_\tau(\sk)$ at $\sk=\ii\cot\tau$. So for $\cot\tau<0$ this pole 
lies in the lower $\sk$-half-plane {\it (the second physical sheet)}. Then there is no 
bound state and $-\Delta_\tau\ge 0$. Conversely $\cot\tau>0$ gives 
rise to a pole of $S(\sk)$ in the upper 
half-plane at $\sk=\ii \cot\tau$ and correspondingly there is one bound state 
with (normalized and real) bound state wave function
\begin{equation}\label{psib}
\psi_{b,\tau}(x)=\sqrt{2\cot\tau}\, \e^{-\cot\tau\, x}
\end{equation}
and with bound state energy
\begin{equation}\label{boundenergy}
\varepsilon_\tau=-\cot^2\tau<0.
\end{equation}
As a consequence $-\Delta_\tau\ge \varepsilon_b$.
Observe that $S_\tau(\sk)$ is real on the imaginary axis, as should be by 
Remark \ref{reality}. Note also agreement with Lemma \ref{lem:mathfracs} and Proposition 
\ref{prop:spectrum}.
In fact in the present case $AB^\dagger=\cos\tau\sin\tau=\cot\tau\sin^2\tau$.

By our general discussion the improper eigenfunctions in this example are given as 
\begin{equation}\label{robineigen}
\psi_\tau(x;\sk)=\e^{-\ii\sk x}+ S_\tau(\sk)\e^{\ii \sk x},\qquad \sk>0. 
\end{equation}
This set is complete if $\cot\tau<0$ while for $\cot\tau>0$ this set combined with 
the bound state wave function \eqref{psib} forms a complete set.
For finite $\cot\tau\neq 0$ and with the condition $m>\max(0,\cot\tau)$, such that 
$\omega(\ii\cot\tau)>0$ for the mass, we obtain $(0<x,y)$
\begin{align}\label{robingreen1}
G(t,x;s,y)&=-\Delta(t-s,x-y;m)-\Delta(t-s,x+y;m)-\fd(t-s, x+y;m,\cot\tau)\\\nonumber
&\qquad\quad+\Theta(\cot\tau){2\cot\tau}\frac{\sin(\omega(\ii\cot\tau)\cdot(t-s))}{\omega(\ii\cot\tau)}
\e^{-\cot\tau(x+y)}.
\end{align}
$\Theta$ is the Heaviside step function. 

This example also provides a nice illustration to a long standing 
problem, namely to what extent the scattering matrix is determined by the cross section 
\cite{Chrichton,IzMartin, MartinI,MartinII,MartinIII, Newton}. Define the scattering amplitude 
$T_\tau(\sk)$ by $S_\tau(\sk)=1+2\ii\, T_\tau(\sk)$, that is 
\begin{equation}\label{scattamp} 
T_\tau(\sk)=-\ii\frac{\cos\tau}{\cos\tau+\ii\sk\sin\tau}.
\end{equation}
The knowledge of $\left|T_\tau(\sk)\right|^2$ for all $\sk>0$ only fixes $\sin^2\tau$. An additional 
information, namely whether there is a bound state or not, is needed to fix $\tau$ itself.
A way to overcome this dilemma and 
to solve this inverse problem in the present context of quantum graphs has been proposed in
\cite{KS11}.

The next example is the single vertex graph with two external lines which may also be viewed as the 
real line with the origin as a distinguished point. 
As boundary conditions we take the the one describing the $\delta$-potential of 
strength $\lambda$ at the origin. This is a very popular model for describing a pointlike impurity.
\begin{example}(The single vertex graph with two external edges ($n=|\cE|=2$) 
and with a boundary condition describing the $\delta-potential$ on the line)\label{ex:delta}~~~\\ 
The graph is obtained by considering two copies of $\R_+$ with their origins identified. 
The real boundary conditions are given as 
$$
A=\begin{pmatrix}1&-1\\0&\lambda\end{pmatrix},\quad B=\begin{pmatrix}0&0\\1&1\end{pmatrix}.
$$
The choices $\lambda<0$ and $\lambda>0$ describe an attractive and a repulsive $\delta$-potential 
on $\R$ respectively.
\end{example}
The resulting on-shell scattering matrix is a symmetric $2\times 2$ matrix  
\begin{equation}\label{slambda}
S_\lambda(\sk)=\frac{1}{2\sk+\ii\lambda}
\begin{pmatrix}-\ii\lambda&2\sk\\2\sk&-\ii\lambda
\end{pmatrix}=\frac{2\sk-\ii\lambda}{2\sk+\ii\lambda}
\;\frac{1}{2}\begin{pmatrix}1&1
\\1&1\end{pmatrix}-\frac{1}{2}\begin{pmatrix}\;1&-1\\
-1&\;1\end{pmatrix}.
\end{equation}
The second expression gives the spectral decomposition \eqref{star1} of the scattering matrix 
for this example, that is $P^0=0$ and 
\begin{equation}\label{slambdaspec}
P^{-\lambda/2}=\frac{1}{2}\begin{pmatrix}1&1
\\1&1\end{pmatrix},\qquad P^\infty=\frac{1}{2}\begin{pmatrix}\;1&-1\\
-1&\;1\end{pmatrix}.
\end{equation}
It has the additional symmetry 
\begin{equation}\label{symmextra}
S_\lambda(\sk)=\begin{pmatrix}0&1\\1&0\end{pmatrix}S_\lambda(\sk)\begin{pmatrix}0&1\\1&0\end{pmatrix}
\end{equation}
describing invariance of the boundary conditions under the interchange of the two edges. 
Using local coordinates we arrange the components $\psi^{l}_j(x;\sk)\, (j,l=1,2;x\,>\,0)$ 
of the two improper eigenfunctions $\psi^{l}(\,;\sk)$ as a $2\times 2$ matrix 
$$
\begin{pmatrix}\e^{-\ii\sk x}-\frac{\ii\lambda}{2\sk+\ii\lambda}\e^{\ii\sk x}
&\frac{2\sk}{2\sk+\ii\lambda}\e^{\ii\sk x}
\\
&\\
\frac{2\sk}{2\sk+\ii\lambda}\e^{\ii\sk x}
&\e^{-\ii\sk x}-\frac{\ii\lambda}{2\sk+\ii\lambda}\e^{\ii\sk x}\end{pmatrix}.
$$ 
Like the S-matrix this matrix is symmetric, reflecting the parity invariance of the 
$\delta$-potential. 
Also ordinary plane waves appear when $\lambda=0$, as they should. The relation \eqref{alphabeta4} 
is easily verified. In the attractive case $\lambda<0$ there is a bound state with bound 
state energy $\epsilon_\lambda=-\lambda^2/4$. The  two local components of the bound state wave function 
are both of the form
\begin{equation}\label{delta3}
\psi_j(x)=\sqrt{-\frac{\lambda}{2}}\e^{\frac{\lambda x }{2}}\quad j=1,2.
\end{equation}
Observe that 
$$
AB^\dagger=\begin{pmatrix}0&0\\0&\lambda\end{pmatrix}
$$
and recall again  Lemma \ref{lem:mathfracs} and Proposition \ref{prop:spectrum} concerning 
the number of bound states.

With the notational convention \eqref{convention} $G_{A,B,m^2}$ in local coordinates can be written 
as a $2\times 2$ matrix in the form, see \eqref{sprop6}, 
\begin{align}\label{delta4} 
G(t,x;s,y)&=-\Delta(t-s,x-y;m)\begin{pmatrix}1&0\\0&1\end{pmatrix}
-\Delta(t-s,x+y;m)\begin{pmatrix}0&1\\1&0\end{pmatrix}\\\nonumber
&\qquad\qquad\qquad-\fd(t-s,x+y;m,-\lambda/2)\;\frac{1}{2}\begin{pmatrix}1&1\\1&1\end{pmatrix}
\\\nonumber
&\qquad
+\Theta(-\lambda)(-\lambda)\frac{\sin(\omega(-\ii\lambda/2)(t-s))}{\omega(-\ii\lambda/2)}
\e^{\frac{\lambda(x+y)}{2}}\frac{1}{2}\begin{pmatrix}1&1\\1&1\end{pmatrix}.
\end{align}

\subsection*{Acknowledgments.}
The author wants to thank L. Faddeev, M. Karowski, V. Kostrykin, and A. Sedrakyan 
for stimulating and helpful comments.


\begin{appendix}
\section{Proof of Relation \eqref{ortho}}\label{app:orthoproof}
\renewcommand{\theequation}{\mbox{\Alph{section}.\arabic{equation}}}
\setcounter{equation}{0} 
The general idea of proof follows a familiar route, see, e.g. \cite{Titchmarsh}. However, the
boundary conditions defining the Laplacian enter in a simple but crucial way, which warrant a more detailed 
discussion. In addition the regularity of the scattering matrix $S(\sk)$ for $\sk>0$ away from $\Sigma^>$ 
will be used.  
For given $R>0$, let $\cG_R$ be the set obtained 
from $\cG$ by deleting from any external edge $e$ all points with distance larger 
than $R$ from its initial vertex $v_e$. On each edge $e$ we introduce an extra vertex at distance $R$ from 
$v_e$ and denoted by $v_{e,R}$. Obviously $\cG_R$ is a compact graph and a closed subset of $\cG$. 
In particular $\cG_R$ has no external edges and hence is compact. 
The set of 
vertices of $\cG_R$ is given as 
$$
\cV_{\cG_R}=\cV\cup\{v_{e,R}\}_{e\in\cE}.
$$
In other words, $\cG_R$ is obtained from $\cG$ by removing the external edges $e\in\cE$, each isomorphic to 
the half-line $[0,\infty)$, and replacing each of them by a closed interval of the form $[0,R]$, where 
the vertex $v_e=\partial(e)$ corresponds to $0\in [0,R]$ and the new vertex $v_{e,R}$ to $R\in [0,R]$.
Correspondingly there is a Hilbert space $L^2(\cG_R)$ with 
scalar product denoted by $\langle\,,\,\rangle_R$. By restriction any function $f$ on 
$\cG$ defines a function on $\cG_R$ also denoted by $f$. In this way any element in 
$L^2(\cG)$ defines an element in $L^2(\cG_R)$ and  
$$
\lim_{R\rightarrow \infty }\langle f,g\rangle_R=\langle f,g\rangle
$$
clearly holds for any $f,g\in L^2(\cG)$. As for the claim \eqref{ortho}, the functions 
$\psi^l(\,;\sk)$ 
are elements in each $L^2(\cG_R)$ but not of $L^2(\cG)$, as already mentioned. 
Now we write 
\begin{equation*}
\langle \psi^l(\,;\sk),\psi^{l^\prime}(\,;\sk^\prime)\rangle_R=-
\frac{1}{\sk^2-\sk^{\prime\,2}}
\left(\langle\Delta_{A,B}\psi^l(\,;\sk),
\psi^{l^\prime}(\,;\sk^\prime)\rangle_R
-\langle\psi^l(\,;\sk),\Delta_{A,B}
\psi^{l^\prime}(\,;\sk^\prime)\rangle_R\right).
\end{equation*}
and perform a partial integration. Since the functions $\psi^l(\,;\sk)$ satisfy the 
boundary conditions, what remains are only contributions from $\psi^l(\,;\sk)$ and its 
first derivative at the vertices $v_{e,R}$.
We now observe
\begin{align*}
\psi^l(v_{e,R};\sk)&=\psi^l_e(x=R;\sk)=\e^{-\ii R\sk}\delta_{l e}
+S(\sk)_{el}\e^{\ii R\sk }\\\nonumber
\frac{d}{dx}\psi^l(v_{e,R};\sk)&=\frac{d}{dx}\psi^l_e(x=R;\sk)=
-\ii \sk \e^{-\ii R\sk}\delta_{l e}+\ii \sk S(\sk)_{el}\e^{\ii R\sk}
\end{align*}
and obtain 
\begin{align*}
\langle \psi^l(\,;\sk),\psi^{l^\prime}(\,;\sk^\prime)\rangle_R
&=-\frac{1}{\sk^2-\sk^{\prime\,2}}\sum_{e\in\cE}
\left(\left(\ii \sk \e^{\ii R\sk}\delta_{l e}-\ii \sk \overline{S(\sk)_{el}}
\e^{-\ii R\sk }\right)
\left(\e^{-\ii R\sk^\prime }\delta_{l^\prime e}
+S(\sk^\prime )_{el^\prime}\e^{\ii R\sk ^\prime }\right)\right.\\\nonumber
&\qquad\qquad\qquad
\left.-\left(\e^{\ii R\sk}\delta_{l e}+\overline{S(\sk)_{el}}\e^{-\ii R\sk }\right)
\left(-\ii \sk^\prime  \e^{-\ii R\sk^\prime }\delta_{l^\prime e}
+\ii \sk^\prime  S(\sk^\prime )_{el^\prime}\e^{\ii R\sk^\prime }\right)\right)\\\nonumber
&=-\frac{\ii}{\sk-\sk^{\prime}}
\left(\e^{\ii R(\sk-\sk^\prime)}\delta_{ll^\prime}
-\sum_{e\in\cE}\overline{S(\sk)_{el}}S(\sk^\prime )_{el^\prime}
\e^{-\ii R(\sk-\sk^\prime)}\right)\\\nonumber
&\qquad\qquad+\frac{\ii}{\sk+\sk^{\prime}}\left(\overline{S(\sk)_{l^\prime l}}
\e^{-\ii R(\sk+\sk^\prime)}
+S(\sk^\prime )_{l l^\prime}\e^{\ii R(\sk+\sk^\prime)}\right).
\end{align*}
Since $\sk+\sk^{\prime}>0$ the second term on the r.h.s. vanishes for $R\rightarrow\infty$ in the sense of 
distributions by the Riemann-Lebesgue lemma. As for the first term write
\begin{align}\label{est1}
-\frac{\ii}{\sk-\sk^{\prime}}
\Big(\e^{\ii R(\sk-\sk^\prime)}\delta_{ll^\prime}&
-\sum_{e\in\cE}\overline{S(\sk)_{el}}S(\sk^\prime )_{el^\prime}
\e^{-\ii R(\sk-\sk^\prime)}\Big)
=\\\nonumber&2\frac{\sin R(\sk-\sk^\prime)}{\sk-\sk^\prime}\delta_{ll^\prime}
-\frac{\ii}{\sk-\sk^{\prime}}\left(\delta_{ll^\prime}-
\sum_{e\in\cE}\overline{S(\sk)_{el}}S(\sk^\prime )_{el^\prime}\right) 
\e^{-\ii R(\sk-\sk^\prime)}.
\end{align}
Here the first term converges in the sense of distributions 
to $2\pi \delta(\sk-\sk^\prime)\delta_{ll^\prime}$ as $R\rightarrow\infty$. 
As for the second term we use the unitarity of $S(\sk)$ to write 
$$
\delta_{ll^\prime}
-\sum_{e\in\cE}\overline{S(\sk)_{el}}S(\sk^\prime )_{el^\prime}
=\sum_{e\in\cE}\overline{S(\sk)_{el}}
\left(S(\sk)_{el^\prime}-S(\sk^\prime )_{el^\prime}\right)
$$
By Corollary \ref{cor:sigma}
all matrix elements of $S(\sk)$ are differentiable functions of 
$\sk\in\R_+\setminus\Sigma^>$.
Since all matrix elements also are bounded by 1 due to unitarity, we have the estimate
$$
|S(\sk)_{ll^\prime}-S(\sk^\prime)_{ll^\prime}| \le 
const\cdot  |\sk-\sk^\prime|,\quad \sk,\sk^\prime\in \R_+\setminus\Sigma^>,\quad 
l,l^\prime\in\cE
$$ 
whenever $|\sk-\sk^\prime|$ is small. Observe that $\R_+\setminus\Sigma^>$ is a
union of open, pairwise disjoint intervals.
This gives the estimate 
$$
\frac{|\delta_{ll^\prime}
-\sum_{e\in\cE}\overline{S(\sk)_{el}}S(\sk^\prime )_{el^\prime}|}
{| \sk-\sk^\prime|}
\le const, \quad \sk,\sk^\prime\in \R_+\setminus\Sigma^>,\quad l,l^\prime\in\cE,
$$ 
again whenever $|\sk-\sk^\prime|$ is small.
Therefore and again by the Riemann-Lebesgue lemma the second term in \eqref{est1} 
tends to zero as $R\rightarrow \infty$.



\section{Proof of Theorem \ref{theo:mfinite} }\label{app:KGsupport}
\renewcommand{\theequation}{\mbox{\Alph{section}.\arabic{equation}}}
\setcounter{equation}{0} 
 
\subsection{Proof of Theorem  \ref{theo:mfinite} in the single vertex case }\label{theo:single}
In the single vertex case, besides a proof of the theorem, in this appendix we will provide a 
detailed analysis of the Klein-Gordon kernel when written in local coordinates, see the 
convention \eqref{convention}. We obtain
\begin{align}\label{starcommutator}
G_{ij}(t,x;s,y)&=-\Delta(t-s,x-y;m)\delta_{ij}+
\int_{-\infty}^\infty\frac{d\sk}{2\pi}S(\sk)_{ij}\e^{\ii\sk(x+y)}\frac{\sin (\omega(\sk)(t-s))}{\omega(\sk)}
\\\nonumber
&\qquad\qquad +\sum_{0<\kappa\in\fI_0}2\kappa P^\kappa_{ij}\e^{-\kappa(x+y)}
\frac{\sin(\omega(\ii\kappa)(t-s))}{\omega(\ii\kappa)}
\end{align}
where we used Corollary \ref{cor:real} and relations \eqref{psimunu}, \eqref{sigma1}, and 
\eqref{diracprep}. Recall also the convention \eqref{mconv} for the case $m=0$.
We can rewrite this as 
\begin{align}\label{sprop6}
G_{ij}(t,x;s,y)&=-\Delta(t-s,x-y;m)\delta_{ij}
-\Delta(t-s,x+y;m)\left(\delta_{ij}-2P^\infty_{ij}\right)\\\nonumber
&\qquad\qquad-\sum_{\infty\neq\kappa\in\fI}\fd(t-s,x+y;m,\kappa)P^\kappa_{jl}\\\nonumber
&\qquad\qquad\qquad+\sum_{0<\kappa\in\fI_0}2\kappa
P^\kappa_{jl}\e^{-\kappa(x+y)}\frac{\sin(\omega(\ii\kappa)(t-s))}{\omega(\ii\kappa)}.
\end{align}
Note that $P^\infty$ may be the zero matrix. We shall use the representation \eqref{starcommutator} to 
prove the theorem.

For the single vertex graph the distance between two points $p$ and $q$ with local 
coordinates $(i,x)$ and $(j,y)$ is 
\begin{equation}\label{stardistance}
d(p,q)=d((i,x),(j,y))=\begin{cases}\,|x-y|\quad &i=j\\
\;\, x+y \quad &i\neq j.
\end{cases}
\end{equation}
As a consequence the first term on the r.h.s. of \eqref{starcommutator}
vanishes for space-like separations, a well known property of the 
relativistic commutator function. As for the integral in \eqref{starcommutator}
insert the relation \eqref{star1}. We observe that 
$d((i,x),(j,y))\le x+y$ is always valid, so for space-like separations $x+y>|t-s|$ holds and thus
we can deform the integral from $-\infty$ to $+\infty$ to the integral from $-\infty+\ii\rho$ to 
$+\infty+\ii\rho$ for arbitrary $\rho>0$. Indeed, by the analyticity of the
first function in \eqref{fanal} we can apply Cauchy's theorem.
During this deformation we pick up a residue at each of the poles $\sk=\ii\kappa$ 
with $0<\kappa <\rho$. Each such term, however,
is compensated by the corresponding term in the sum in \eqref{starcommutator}. 
When we let $\rho\rightarrow +\infty$, we claim that the integral from $-\infty+\ii\rho$ to 
$+\infty+\ii\rho$ vanishes. To see this,view the function $\sk\mapsto \omega(\sk)$ as analytic in the cut 
(open) upper $\sk$-half-plane with a cut from $\ii m$ to $\ii\infty$, see Figure \ref{cutplane0}.
In this cut upper $\sk$-half-plane, the estimate $\Im\, \omega(\sk)\le \Im\, \sk$ holds. 

\begin{figure}[htb]
\setlength{\unitlength}{0.8cm}
\begin{picture}(12,6)
\put(0,0){\line(1,0){10}}
\put(5,0){\line(0,1){2.5}}
\put(9,0.3){Re $\sk$}
\put(5.2,5.7){Im $\sk$}
\put(5,2.5){\circle*{0.17}}
\put(4.2,2.2){$\ii m$}
\put(5,3.0){\oval(1,1)[b]}
\put(4.5,3.0){\line(0,1){2.5}}
\put(5.5,3.0){\line(0,1){2.5}}
\put(3.8,3.5){$\cC_-$}
\put(6.0,3.5){$\cC_+$}
\thicklines
\put(5.0,2.5){\line(0,1){3.0}}
\end{picture}
\caption{\label{cutplane0} The upper $\sk$-half-plane with a cut from $\ii m$ to 
$\ii\infty$ for the function $\omega(\sk)$.
$\cC_\pm$ form the lips of the cut.}
\end{figure}
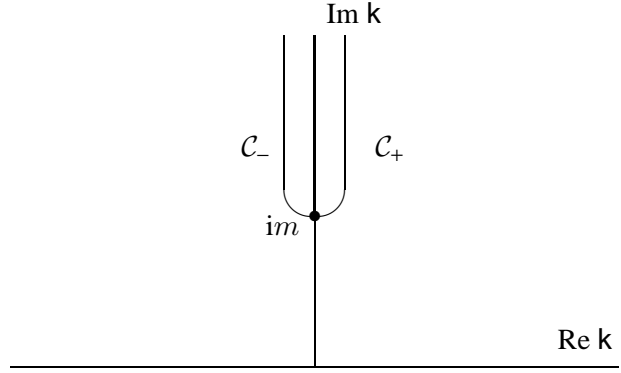
Moreover both functions
\begin{equation}\label{decomp}
\frac{1}{2\ii\omega(\sk)}\e^{\ii\omega(\sk)(t-s)},\qquad -\frac{1}{2\ii\omega(\sk)}\e^{-\ii\omega(\sk)(t-s)}
\end{equation}
are also analytic there and their sum equals 
$$
\frac{\sin(\omega(\sk)(t-s))}{\omega(\sk)}
$$
there. Furthermore this sum has no discontinuity across the cut, as it should since it is entire analytic. 
Indeed, replace $\sk$ by the variable 
$m\le\lambda< \infty$ via $\sk=\ii\lambda-\epsilon$ 
on the left lip $\cC_-$ and $\sk=\ii\lambda+\epsilon$ on the right lip $\cC_+$ with $\epsilon>0$. 
But on the left lip
$$
\lim_{\epsilon\downarrow 0}\omega(\ii\lambda-\epsilon)=\sqrt{\lambda^2-m^2}
$$
while on the right lip 
$$
\lim_{\epsilon\downarrow 0}\omega(\ii\lambda+\epsilon)=-\sqrt{\lambda^2-m^2}.
$$ 
Using $\Im\, \omega(\sk)\le \Im\, \sk$ for $\sk$ in the upper half plane we can therefore estimate 
\begin{equation}\label{kernelest}
\left|\frac{\sin(\omega(\sk)(t-s))}{\omega(\sk)}\right|\le\frac{e^{\Im \sk |t-s|}}{|\omega(\sk)|}.
\end{equation}
in the upper half plane and which combined with 
$$
\left|\e^{\ii\sk(x+y)}\right|= \e^{-\Im \sk (x+y)}
$$
proves the claim.
This concludes the proof of the Theorem \ref{theo:mfinite} when the graph is a single vertex graph.
Observe that we have actually proved
\begin{equation}\label{local0}
G_{ij}(t,x;s,y)=-\Delta(t-s,x-y;m)\delta_{ij}\qquad \mbox{when}\quad x+y>|t-s|.
\end{equation}
\eqref{sprop6} compares with \eqref{local0}, valid when 
$x+y>|t-s|$. If at least one of the points $p$ and $q$ is far away from the vertex, 
that is $x \gg 1$ or $y\gg 1$, then the last term on the r.h.s. of \eqref{sprop6} becomes exponentially 
small, uniformly for all times $t$ and $s$.
To sum up, as far as commutators are concerned and by comparison with \eqref{jordanfunction}, 
the contribution 
from $\fd$ in \eqref{sprop6} compares with the two preceding terms there. 
\begin{remark}\label{re:loc}
We observe from the proof that in the single vertex case the bound state 
contributions in the definition of the fields are necessary in order to obtain locality.
A somewhat similar observation was made in the context of integrable models 
in quantum field theory \cite{KarowskiWeisz}. There it was observed that bound
state contributions in the form factors of the Sine-Gordon model were crucial for
determining the wave-function renormalization constant. Moreover, in the articles 
\cite{BFK,Q} local commutation relations for certain integrable models were established,
see in particular relation (54) in \cite{BFK}, for which also contributions from bound states 
are relevant.
\end{remark}

\subsection{Proof of Theorem  \ref{theo:mfinite}  for an arbitrary graph when $\Sigma=\emptyset$}
\label{theo:general}
We turn to the case of an arbitrary graph with the spectral assumption $\Sigma=\emptyset$ for the 
Laplacian $-\Delta_{A,B}$, that is with the assumption that there are no bound states. 
In local coordinates 
\begin{equation}\label{gencommutator}
G_{ij}(t,x;s,y)=\begin{cases}-\Delta(t-s,x-y;m)\delta_{ij}+\int_{-\infty}^\infty\frac{d\sk}{2\pi}
S(\sk)_{ij}\e^{\ii\sk(x+y)}\;\frac{\sin(\omega(\sk)(t-s))}{\omega(\sk)}\\
\hspace{8.3cm}\mbox{for}\quad i,j\in\cE\\
\\
\int_{-\infty}^\infty\frac{d\sk}{2\pi}\left(\alpha(\sk)_{ij}\e^{\ii\sk(x+y)}
+\beta(\sk)_{ij}\e^{\ii\sk(-x+y)}\right)\;\frac{\sin(\omega(\sk)(t-s))}{\omega(\sk)}\\\hspace{8.3cm}
\mbox{for}\quad i\in\cI,j\in\cE\\
\\
\int_{-\infty}^\infty\frac{d\sk}{2\pi}
\left(\left(\alpha_{A,B}(\sk)\e^{\ii\sk x}+\beta_{A,B}(\sk)\e^{-\ii\sk x}\right)
\beta_{\bar{A},\bar{B}}(-\sk)^T\right)_{ij}\e^{\ii\sk y} 
\;\frac{\sin(\omega(\sk)(t-s))}{\omega(\sk)}\\\hspace{8.3cm}\mbox{for}\quad i,j\in\cI.
\end{cases}
\end{equation}
Relation \eqref{unitarity} has been used for the case $i,j\in\cE$, Corollary \ref{lem:hermanal} 
for the case $i\in\cI,j\in\cE$. Lemma \ref{lem:hermanal} and Corollary \ref{cor:complex} 
have been used for the case $i,j\in\cI$.
Consider first the case $j,l\in\cE$. The first term, the relativistic commutator 
function, has already been dealt with and 
vanishes for space-like separations. As for the integral we insert the path space expansion \eqref{walks} 
for the scattering matrix  to obtain the representation 
\begin{equation}\label{commutator3}
\sum_{\bw\in \cW_{ij}}
\int_{-\infty}^\infty\frac{d\sk}{2\pi}
S(\bw;\sk)_{ij}\e^{\ii\sk(x+y+|\bw|)}\frac{\sin(\omega(\sk)(t-s))}{\omega(\sk)} .
\end{equation}
Here and in what follows we will freely interchange summation and integration. This is permitted as can 
be shown with help of Proposition 5.6 in \cite{KS8} and where one lets the lengths $a_i$ of the internal 
edges become complex with a positive imaginary part. We omit details.

For events, which are space-like separated, 
$ x+y+|\bw| \ge d((i,x),(j,y))> |t-s| $ is valid for any 
$\bw\in \cW_{ij}$ whenever $i,j\in\cE$. Also by Lemma \ref{lem:mathfracs}, the 
assumption $\Sigma=\emptyset$ implies $A B^\dagger\le 0$ which in turn implies 
that $A(v) B(v)^\dagger\le 0$ holds for all vertices $v$ by Lemma \ref{lem:decomp}. 
This in turn implies that each $S(v;\sk)$, which is of the form 
$-(A(v)+\ii\sk B(v))^{-1}(A(v)-\ii\sk B(v))$, has no poles and 
and hence is analytic in the upper half plane and polynomially bounded there, again by 
Lemma \ref{lem:mathfracs}. As a consequence each $S(\bw;\sk)_{jl}$ is analytic 
in the upper half-plane and polynomially bounded. 
These considerations again allow us to make a deformation of the 
integration over $\sk$ in \eqref{commutator3} from the real axis $(-\infty,+\infty)$ to the parallel line 
$(-\infty+\ii\rho,+\infty+\ii\rho)$. Combining the estimate \eqref{kernelest} with 
$$
\left|\e^{\ii\sk(x+y+|\bw|)}\right|= \e^{-\Im \sk ((x+y+|\bw|)}
$$
and the polynomial bound of each $S(\bw;\sk)_{jl}$ in the limit $\rho\rightarrow +\infty$ 
we obtain a vanishing contribution. In other words, each summand in \eqref{commutator3} vanishes.
This concludes our discussion of the case $i,j\in\cE$.

We turn to the case $i\in\cI$ and $j\in\cE$ and discuss the integral involving the 
$\alpha(\sk)$ and $\beta(\sk)$ amplitudes separately. By the walk expansion \eqref{absumrep}
\begin{align}
\int_{-\infty}^\infty\frac{d\sk}{2\pi}
\alpha(\sk)_{ij}\e^{\ii\sk(x+y)}\frac{\sin(\omega(\sk)(t-s))}{\omega(\sk)} 
&=\sum_{\bw\in \cW_{ij}^-}\int_{-\infty}^\infty \frac{d\sk}{2\pi}
S(\bw;\sk)_{ij}\e^{\ii\sk(x+ |\bw|+y)}\frac{\sin(\omega(\sk)(t-s))}{\omega(\sk)} 
\end{align}
and we observe that $d((j,x),(l,y))\le x+ |\bw|+y $, holds for all 
$\bw\in \cW_{jl}^-$, see \eqref{dini}. Hence for space-like separation and for each summand 
we can again deform the integration contour to $(-\infty+\ii\rho,+\infty+\ii\rho)$ 
and thus this expression then vanishes when $\rho\rightarrow \infty$. As for the 
term containing the amplitude $\beta(\sk)$, again the walk expansion gives
\begin{align}
\int_{-\infty}^\infty\frac{d\sk}{2\pi}
\beta(\sk)_{ij}\e^{\ii\sk(-x+y)}\frac{\sin(\omega(\sk)(t-s))}{\omega(\sk)} 
&=\sum_{\bw\in \cW_{ij}^+}\int_{-\infty}^\infty \frac{d\sk}{2\pi}
S(\bw;\sk)_{ij}\e^{\ii\sk(a_i-x+ |\bw|+y)}\frac{\sin(\omega(\sk)(t-s))}{\omega(\sk)} 
\end{align}
Now $d((j,x),(l,y))\le a_i-x+ |\bw|+y$ holds for all 
$\bw\in \cW_{ij}^+$, cf. again \eqref{dini}, and the previous arguments can again be applied. 

In the case $j,l\in\cI$, the arguments just used do not work. This is the reason why we have been 
unable to establish finite propagation speed inside the graph, that is in $\cG_{\mathrm{int}}$. 
Indeed, now the contour deformation into the upper $\sk$-half plane can not be carried out, since 
$\beta_{\bar{A},\bar{B}}(-\sk)$  will have poles in the 
upper half-plane. Also the walk representation of $\beta_{\bar{A},\bar{B}}(-\sk)$ for $\sk>0$ 
does not have the form needed to invoke the arguments we have used so far.

\begin{remark}\label{re:generalbdy}
The reason we had to impose the condition $\Sigma=\emptyset$ for a general graph
is that in the presence of bound states
we do not (yet) have sufficient control over the matrix valued functions $S(\sk),\alpha(\sk)$ 
and $\beta(\sk)$ at the poles. Recall that in the single vertex case we had 
Proposition \ref{prop:starev} at our disposal.
However, we expect Einstein causality still to be valid without this condition. 
\end{remark}
\end{appendix}


\end{document}